\documentclass[11pt]{article}
\usepackage{amssymb, amsmath, amsfonts, eurosym, geometry, graphicx, color, setspace, sectsty, footmisc, caption, natbib, pdflscape, array, hyperref, float, verbatim, listings, xcolor, commath, enumitem, xfrac, pgfplots, mathtools, tikz, multirow, diagbox, chngpage, booktabs, tabularx, dirtytalk, indentfirst, bbm, subcaption, babel, dsfont, aeguill, amsthm, multicol, fancyhdr, tocloft, fontawesome5, wasysym, bookmark}
\usepackage[utf8]{inputenc}
\usepackage[affil-it]{authblk}

\title{Directional Dependence of Extreme Events\thanks{
The complete replication package for this study is publicly available at: \url{https://github.com/maximenc/Directional-Dependence-Extreme-Events}. MG acknowledges the support of the Chaire ``Deep Finance Statistics'' between QRT, Ecole Polytechnique and its foundation. We thank Duc Thi Luu, Jiang Pu, Jean-Gabriel Attali and participants of the Quantitative Finance Workshop at De Vinci Research Center in Paris for helpful suggestions and comments. We also thank Cambyse Pakzad and Carlo Gaetan for suggestions and participants of the 56th Statistical Days of the French Statistical Society. We further thank Tsung-I Lin and the participants of the 8th International Conference on Econometrics and Statistics. \\
\vspace{-0.25cm}
E-mail addresses: \texttt{matthieu.garcin@m4x.org} (M. Garcin), and \texttt{m.nicolas@ucl.ac.uk} (M. L. D. Nicolas)}
}

\date{\today}

\newcommand{\proba}{\mathbb{P}} 
\newcommand{\indic}{\mathds{1}} 
\newcommand{\E}{\mathbb{E}} 

\newtheorem{definition}{Definition}[section]

\newtheorem{theorem}{Theorem}[section]

\newtheorem{proposition}{Proposition}[section]

\newcommand{\indep}{\raisebox{0.05em}{\rotatebox[origin=c]{90}{$\models$}}}

\begin{document}

\author[1]{Matthieu Garcin}
\author[2]{Maxime L. D. Nicolas}
\affil[1]{\small{\textit{De Vinci Higher Education, De Vinci Research Center, Paris, France}}}
\affil[2]{\small{\textit{Institute of Finance \& Technology, University College London, London, UK}}}

\maketitle

\vspace{-0.5cm}

\begin{abstract}
\linespread{1.1}\selectfont
\noindent

This paper introduces a novel measure to quantify the directional dependence of extreme events between two variables. The proposed approach is designed to capture asymmetric tail dependence by studying conditional tail expectations of rank-transformed variables, thereby quantifying the behavior of one variable when the other takes extreme values. We investigate the theoretical asymptotic behavior of the associated estimator. The effectiveness of the approach is demonstrated through an extensive simulation study. In addition, we discuss the use of the proposed coefficient for the detection of causal effects in extreme events.
Finally, we apply the method to an oceanographic dataset, where the results highlight the strong asymmetric nature of extreme events and identify the dominant directions of extremal influence among key oceanographic variables.
As a directional measure of tail dependence, our approach provides a natural tool for exploring causal-effect relationships in extreme-value settings.


\end{abstract}


\newpage

\section{Introduction}

Extreme events often propagate asymmetrically across systems. In meteorology, \cite{jasson2005asymetrie} documents that storms in Europe often propagate from west to east, so extreme events affecting Germany are likely to have already impacted France, whereas the reverse propagation is less common. In hydrology, extreme rainfall upstream often induces downstream flooding, whereas extreme downstream water levels may occur without corresponding upstream extremes \citep{deidda2023asymmetric}. In finance, severe market downturns typically induce extreme losses in individual assets, whereas extreme movements in a single asset may occur without triggering a market-wide crash. Similar asymmetries arise in seismology, energy systems, and engineering networks, where rare shocks may be transmitted strongly in one direction but weakly---or not at all---in the opposite direction.

Traditional measures of tail dependence focus on the joint occurrence of extreme events and do not capture directional effects. In particular, the tail dependence coefficient (TDC) and its variants are symmetric by construction \citep{sibuya1960bivariate, joe1997multivariate} and therefore provide no information on the behavior of isolated extremes, that is, extreme events that occur in one variable without a corresponding extreme in the other. Consider a situation in which extreme events of a variable $X$ always occur jointly with extremes of another variable $Y$, while extreme events of $Y$ occur independently of extremes of $X$. Consider, conversely, the opposite configuration. These two fundamentally different dependence structures yield identical TDCs, as this measure only captures the probability of joint extremes. A meaningful measure of directional tail dependence should therefore capture how extreme or how non-extreme a variable is on average when the other is extreme, thereby accounting for the presence of isolated extreme events.

To address this limitation of TDCs, several authors have proposed measures based on conditional tail behavior through conditional tail expectations (CTEs), such as $\mathbb{E}[X \mid Y>t]$ or $\mathbb{E}[X \mid Y=t]$, as proposed by \cite{hua2014strength, hua2019assessing}. These quantities aim to capture how extreme realizations of $Y$ affect the behavior of $X$. However, CTEs depend on both the marginal distributions and the dependence structure, which makes meaningful comparisons across variables difficult. A related idea underlies the \textit{tail exchangeability} index introduced by \cite{hua2019assessing}, but this approach requires identically distributed non-negative variables, substantially limiting its applicability.

Other recent approaches define directional tail dependence through alternative frameworks. \cite{gnecco2021causal} propose a causal tail coefficient within heavy-tailed linear structural causal models, focusing explicitly on causal inference but relying on strong modeling assumptions regarding the tail of the distribution and on the linear dependence between variables. \cite{jiang2024efficient} introduce a notion of directional tail dependence based on the geometry of tail regions and extend this framework to extremal causality analysis \citep{li2024multivariate}. While these contributions provide valuable insights, they either impose specific structural assumptions, lack marginal invariance, or do not yield simple empirical estimators that can be readily implemented in general settings. More generally, necessary properties for directed dependence measures have been discussed by \citep{junker2021estimating}.



In this paper, we propose a new measure of directional tail dependence that directly targets asymmetric extremal behavior while remaining marginal-free. Our approach is based on transforming each variable using its probability integral transform and studying the expected transformed value of one variable conditional on the other taking extreme values. This rank-based transformation places all variables on a common scale and allows extremal behavior to be compared across directions independently of the marginal distributions. The resulting quantity admits a natural expression in terms of the copula, thereby isolating the dependence structure.

Beyond its definition, we develop an inferential framework for the proposed measure. We construct an empirical estimator based on the empirical copula and establish its strong consistency and asymptotic normality under mild smoothness conditions \citep{segers2012asymptotics,remillard2009testing}. These results allow formal statistical testing of directional extremal effects. We illustrate the finite-sample performance of the estimator through simulation studies using asymmetric copula constructions, including Khoudraji's device \citep{khoudraji1997contributions,genest2012tests}, and apply the methodology to oceanographic and hydrological data, where directional asymmetries in extreme events are well documented \citep{horning2025simulation}.

The proposed measure quantifies directional extremal dependence rather than causality per se. While asymmetric tail behavior may reflect underlying propagation mechanisms in specific applications, the coefficient is defined purely in terms of the joint distribution and does not rely on structural or temporal assumptions. It therefore provides a robust descriptive tool for identifying and quantifying directional features of extreme events.

The rest of the paper is organized as follows. Section~\ref{sec:dir-tail-dep} introduces the concept of directional tail dependence and defines the proposed coefficient based on transformed conditional tail expectations. Section~\ref{sec:asymp-prop} presents the asymptotic properties of the estimator, including consistency and limit distribution. Section~\ref{sec:sim-study} presents a simulation study, with two main setups: Khoudraji's device and the skew-t copula. Section~\ref{sec:data-application} applies the method to oceanographic data, and Section~6 concludes.

\section{Directional Tail Dependence } \label{sec:dir-tail-dep}

To quantify asymmetries in the behavior of two variables under extreme conditions, we propose a measure that compares the behavior of $X$ when $Y$ is extreme with the behavior of $Y$ when $X$ is extreme. Our approach builds on the marginal expected shortfall (MES), also referred to as the conditional tail expectation (CTE).\footnote{Although closely related, the CTE is often defined in a univariate setting. It is equivalent to the Tail Value at Risk (CVaR) and corresponds to the expected value of $X$ conditional on $X$ exceeding its $(1-p)$-th quantile, that is, $\mathbb{E}[X \mid X > Q_X(1-p)]$, where $Q_X$ denotes the quantile function of $X$.} In a bivariate context, this concept naturally extends to quantities such as $\mathbb{E}[X \mid Y \le t]$ or $\mathbb{E}[X \mid Y = t]$ for small values of $t$ \citep{cai2005conditional}. These expressions describe the expected behavior of $X$ when $Y$ attains extreme (low) values.

However, the limitation of this approach is that it depends on the marginal distributions of $X$ and $Y$, making it difficult to compare $\mathbb{E}[X \mid Y \leq t]$ with $\mathbb{E}[Y \mid X \leq t]$ when $X$ and $Y$ have different scales or distributional shapes. To address this, a transformation is required that places both variables on a common scale. In particular, meaningful comparisons between bivariate random vectors with arbitrary marginals are only possible once the marginals are standardized into the same Fréchet (or uniform) class \citep{dette2013copula}.

Therefore, to allow a comparison between the tail behaviors of two random variables with different marginal distributions, we apply the probability integral transform to standardize each variable. Specifically, we transform $X$ and $Y$ using their respective marginal cumulative distribution functions $F_X$ and $F_Y$, which are assumed to be continuous. Thanks to this probability integral transform, the new variables $F_X(X)$ and $F_Y(Y)$ are uniformly distributed on the unit interval. This transformation corresponds to scaled ranks in empirical settings and thus allows to compare the extremeness of multiple variables on a common scale. We then introduce the following quantity to assess the behavior of $X$ conditionally to the lower tail of $Y$.

\begin{definition}[Transformed Conditional Tail Expectation]
The Transformed Conditional Tail Expectation (TCTE) of $X$ given $Y$ is defined, for $v \in (0,1)$, as
$$\chi^{Y \rightarrow X}(v)
= \mathbb{E}\!\left[ F_X(X) \,\middle|\, F_Y(Y) \le v \right].$$
The corresponding limiting quantity is defined as
$$\chi^{Y \rightarrow X}
= \lim_{v \to 0^+} \chi^{Y \rightarrow X}(v),$$
provided the limit exists.
\end{definition}

The TCTE can be interpreted as the expected scaled rank of one variable conditional on the other attaining extreme values.%
\footnote{A related quantity, termed the \emph{causal tail coefficient}, is introduced by \cite{gnecco2021causal} within a linear structural causal model framework \citep{pearl2009causality}. While their objective is causal identification, our approach focuses instead on characterizing asymmetric dependence in the tails using the copula, without imposing a specific causal structure.}
By construction, the transformation by the marginal distribution functions yields a scale-invariant measure that relies only on the dependence structure. In particular, the TCTE admits a natural representation in terms of the copula associated with the joint distribution of $(X,Y)$.

\begin{proposition}\label{proposition:dtdc-copula}
Let $C$ denote the copula of $(X,Y)$, their marginal cumulative distribution function being continuous. Then, for $v \in (0,1)$,
\begin{equation}\label{eq:tcte_copula}
\chi^{Y \rightarrow X}(v)
= 1 - \frac{1}{v} \int_0^1 C(u,v)\, \mathrm{d}u.
\end{equation}
Moreover, if the copula $C$ is differentiable in its second argument at $v=0$, then
\begin{equation}\label{eq:tcte_limit_copula}
\chi^{Y \rightarrow X}
= 1 - \int_0^1 \partial_2 C(u,0)\, \mathrm{d}u,
\end{equation}
where $\partial_2 C(u,0)$ denotes the partial derivative of $C$ with respect to its second argument evaluated at $v=0$. 
\end{proposition}

The proof of Proposition~\ref{proposition:dtdc-copula} is postponed in Appendix~\ref{sec:proof_copula_form}.

The value of $\chi^{Y \rightarrow X}$ provides a natural interpretation in terms of directional tail influence. If $\chi^{Y \rightarrow X} < \tfrac{1}{2}$, we have that $Y$ exerts a \emph{positive tail influence} on $X$, meaning that extremely low realizations of $Y$ tend to be associated with relatively low realizations of $X$. Conversely, if $\chi^{Y \rightarrow X} > \tfrac{1}{2}$, this indicates an \emph{anti-dependence} or \emph{negative tail influence} of $Y$ on $X$, whereby extremely low values of $Y$ are typically associated with relatively large values of $X$.

In principle, the function $\chi^{Y \rightarrow X}(v)$ could oscillate around $\tfrac{1}{2}$ for different values of $v$.  To ensure a consistent interpretation of tail influence, we focus on dependence structures that exhibit a global form of positive or negative association in the tails. Specifically, we assume quadrant dependence \citep{denuit2004nonparametric, scaillet2005kolmogorov}, which we recall below.

\begin{definition}[Positive Quadrant Dependence]
A pair of random variables $(X,Y)$ is said to be \emph{positively quadrant dependent} (PQD) if
$$
P(X \le x, Y \le y) \ge P(X \le x) P(Y \le y), \quad \forall x,y,
$$
or, equivalently, if the associated copula $C$ satisfies
\begin{equation}\label{eq:pqd}
C(u,v) \ge uv, \quad \forall (u,v) \in [0,1]^2.
\end{equation}
The pair $(X,Y)$ is said to exhibit \emph{negative quadrant dependence} (NQD) when the inequality is reversed.
\end{definition}

Although quadrant dependence restricts the class of admissible dependence structures, it is a relatively weak assumption that holds for a wide range of dependence structures in practice. From equation~\eqref{eq:tcte_copula}, PQD implies $\chi^{Y\rightarrow X}(v) \leq \frac{1}{2}$ and $\chi^{X\rightarrow Y}(v) \leq \frac{1}{2}$, while NQD implies $\chi^{Y\rightarrow X}(v) \geq \frac{1}{2}$ and $\chi^{X\rightarrow Y}(v) \geq \frac{1}{2}$. The boundary case $\chi^{Y \rightarrow X}(v)=\tfrac{1}{2}$ corresponds to the independence copula $C(u,v)=uv$. However, while independence implies $\chi^{Y \rightarrow X}(v)=\tfrac{1}{2}$, the converse does not necessarily hold.

The potential asymmetry between $\chi^{Y \rightarrow X}$ and $\chi^{X \rightarrow Y}$ motivates the introduction of a directional tail dependence measure.

\begin{definition}[Directional Tail Dependence]\label{def:DTD}
Let $(X,Y)$ be a bivariate random vector with marginal distribution functions $F_X$ and $F_Y$, and let
$U_X = F_X(X)$ and $U_Y = F_Y(Y)$ denote the corresponding probability integral transforms. 
Assume that $(X,Y)$ exhibits PQD or NQD.
Then, the pair $(X,Y)$ is said to exhibit directional tail dependence (DTD) if
\begin{equation}\label{eq:dtd1}
\chi(X,Y) := \chi^{Y \rightarrow X} - \chi^{X \rightarrow Y} \neq 0.
\end{equation}
If $\chi(X,Y)=0$, the pair $(X,Y)$ is said to exhibit non-directed tail dependence. When $\chi(X,Y)\neq 0$ the direction of the tail influence is determined by the sign of $\chi(X,Y)$. We say that:
\begin{itemize}[noitemsep]
    \item under PQD (respectively NQD), $X$ is tail-directed toward $Y$ if $\chi(X,Y) > 0$ (resp. $< 0$).
    \item under PQD (respectively NQD), $Y$ is tail-directed toward $X$ if $\chi(X,Y) < 0$ (resp. $> 0$).
\end{itemize}
\end{definition}

The magnitude of $\chi(X,Y)$ quantifies the strength of asymmetric tail behavior between $X$ and $Y$, while its sign along with the sign of $\chi^{Y \rightarrow X}-1/2$ and $\chi^{X \rightarrow Y}-1/2$ identifies the direction of this asymmetry. Large absolute values of $\chi(X,Y)$ indicate pronounced directional effects in the tails, whereas values close to zero suggest near-symmetric extremal behavior. Under PQD, if $\chi(X,Y) > 0$, extremely low realizations of $X$ tend to be associated with relatively small values of $Y$ more frequently than the reverse situation. This indicates that extremes of $X$ exert a stronger influence on the lower tail behavior of $Y$ than extremes of $Y$ do on $X$. Conversely, under NQD, if $\chi(X,Y) > 0$, extremely low realizations of $Y$ are more likely to coincide with relatively large values of $X$, indicating a stronger extremal influence of $Y$ on the upper tail behavior of $X$. We note that non-reciprocal tail dependences may also appear, in which, for example, $\chi^{Y \rightarrow X}>1/2$ and $\chi^{X \rightarrow Y}<1/2$. In this situation, we cannot interpret $\chi(X,Y)$ in the same way as in Definition~\ref{def:DTD}.

By construction, the directional tail dependence coefficient $\chi(X,Y)$ satisfies the following properties:
\begin{enumerate}[noitemsep]
    \item $\chi(X,Y) \in [-1,1]$, restricted to $\chi(X,Y) \in [-1/2,1/2]$ in case of PQD or NQD (\emph{boundedness});
    \item $\chi(X,Y) = -\chi(Y,X)$ (\emph{antisymmetry});
    \item $X \indep Y \;\Rightarrow\; \chi(X,Y)=0$ (\emph{null under independence});
    \item for the TDC $\lambda_{XY}$, we have $\lambda_{XY}=1 \;\Rightarrow\; \chi(X,Y)=0$ (\emph{null under complete tail dependence});
    \item $(X,Y) \stackrel{d}{=}(Y,X) \;\Rightarrow\; \chi(X,Y)=0$ (\emph{null under exchangeability});
    \item $\chi\!\left(g(X),h(Y)\right)=\chi(X,Y)$ for strictly increasing functions
    $g,h:\mathbb{R}\to\mathbb{R}$ (\emph{scale invariance}).
\end{enumerate}

These properties raise a few remarks.

Under exchangeability, the joint distribution of $(X,Y)$ is symmetric with respect to permutation, implying a null directional tail dependence coefficient. This symmetry property is inherent to many commonly used multivariate models, including most copula families, and directional tail dependence therefore quantifies deviations from this standard modeling assumption.

A value $\chi(X,Y)=0$ does not characterize the strength of extremal dependence between $X$ and $Y$. In particular, a null directional tail dependence may arise both under independence and under perfect extremal dependence ($\lambda_{XY}=1$). In such cases, the tail behavior is either symmetric or perfectly coupled, so that no directional asymmetry can be detected by $\chi(X,Y)$.

Although $\chi(X,Y)$ is entirely determined by the copula, it differs fundamentally from the \emph{tail exchangeability} index of \cite{hua2019assessing}, which is defined for survival functions of non-negative random variables and does not provide a directional comparison between variables.

The definition of $\chi(X,Y)$ focuses on the lower tails of the marginal distributions. An analogous construction can be obtained for upper-tail dependence by replacing the conditioning events with $F_X(X) \ge 1 - v$ and $F_Y(Y) \ge 1 - v$.

The directional tail dependence coefficient $\chi(X,Y)$ captures asymmetries in the joint tail behavior of $(X,Y)$ and is defined purely in terms of their joint distribution. As such, it does not, by itself, establish a causal relationship. However, directional asymmetries in the tails may be informative about potential causal mechanisms or propagation effects in specific applications. The coefficient $\chi(X,Y)$ can therefore be viewed as a descriptive tool that highlights directional extremal patterns and helps identify situations in which a more detailed causal analysis may be needed.

\section{Nonparametric Estimation of TCTE} \label{sec:asymp-prop}

\subsection{Estimator and Consistency}

Let $(X_1,Y_1),\ldots,(X_n,Y_n)$ be an i.i.d.\ sample from the joint distribution of $(X,Y)$. 
For any $v \in (0,1)$, an empirical estimator of $\chi^{Y \rightarrow X}(v)$ can be defined through the empirical copula as
\begin{equation}\label{eq:estimChiInt}
\widehat{\chi}^{Y \rightarrow X}_n(v)
= 1 - \frac{1}{v}\int_0^1 \widehat{C}_n(u,v)\,\mathrm{d}u,
\end{equation}
where the empirical copula $\widehat{C}_n$ is given by
$$
\widehat{C}_n(u,v)
= \frac{1}{n}\sum_{i=1}^n 
\indic_{\{\widehat{F}_{X,n}(X_i)\le u\}}
\indic_{\{\widehat{F}_{Y,n}(Y_i)\le v\}},
$$
and $\widehat{F}_{X,n}$ and $\widehat{F}_{Y,n}$ denote the empirical distribution functions of $X$ and $Y$, respectively.
By discretizing the integral in \eqref{eq:estimChiInt}, the estimator $\widehat{\chi}^{Y \rightarrow X}_n(v)$ can equivalently be expressed as a finite sum, which is more convenient for numerical implementation. Indeed,
$$
\begin{aligned}
\widehat{\chi}^{Y \rightarrow X}_n(v)
&= 1 - \frac{1}{v}\sum_{i=1}^n \int_{(i-1)/n}^{i/n} \widehat{C}_n(u,v)\,\mathrm{d}u \\
&= 1 - \frac{1}{vn}\sum_{i=1}^n \widehat{C}_n\!\left(\frac{i-1}{n},v\right) \\
&= 1 - \frac{1}{vn}\sum_{i=1}^{n-1} \widehat{C}_n\!\left(\frac{i}{n},v\right),
\end{aligned}
$$
where the last equality follows from the fact that $\widehat{C}_n(0,v)=0$. Recalling the definition of $\chi^{Y \rightarrow X}$ as the limit of $\chi^{Y \rightarrow X}(v)$ as $v \to 0^+$, a natural estimator of $\chi^{Y \rightarrow X}$ is obtained by setting
$\widehat{\chi}^{Y \rightarrow X}_n
:= \widehat{\chi}^{Y \rightarrow X}_n(v_n),$
where $(v_n)_{n\ge1}$ is a sequence of positive real numbers converging to zero at an appropriate rate, for instance $v_n=\sqrt{\log(n)/n}$. The asymptotic properties of this estimator are established in the following theorem.

\begin{theorem}\label{thm:Consistent}
For any fixed $v \in (0,1]$, the estimator $\widehat{\chi}^{Y \rightarrow X}_n(v)$ of $\chi^{Y \rightarrow X}(v)$ is strongly consistent. 
Moreover, if the copula $C$ is differentiable with respect to its second argument at $v=0$, uniformly in $u \in [0,1]$, and if $(v_n)_{n\ge1}$ satisfies
$ v_n \to 0
\quad \text{and} \quad
v_n \sqrt{n/\log\log n} \to \infty
\quad \text{as } n \to \infty, $
then $\widehat{\chi}^{Y \rightarrow X}_n(v_n)$ is a strongly consistent estimator of $\chi^{Y \rightarrow X}$.
\end{theorem}
The proof of Theorem~\ref{thm:Consistent} is given in Appendix~\ref{sec:proof_Consistent}.

\subsection{Limit Distribution and Threshold Selection}

We now establish the asymptotic normality of the empirical estimator
$\widehat{\chi}^{Y \rightarrow X}_n(v)$ for fixed $v \in (0,1]$.
This result provides the basis for statistical inference and for assessing the statistical significance of directional tail dependence.

\begin{theorem}\label{thm:GaussChiV}
Let the bivariate copula $C$ have partial derivatives, with $\partial_1C$ continuous on $(0,1)\times[0,1]$ and $\partial_2C$ continuous on $[0,1]\times(0,1)$. Then, for any fixed
$v \in (0,1]$,
$$\sqrt{n}\left(\widehat\chi^{Y\rightarrow X}_n(v)-\chi^{Y\rightarrow X}(v)\right) \overset{d}{\longrightarrow} \mathcal N(0,\sigma^2_C(v)),$$
where 
$$\begin{array}{ccl}
\sigma^2_C(v) & = & \frac{2}{v^2} \int_0^1\int_0^w\left\{ C(u,v)(1-C(w,v))+\partial_2C(u,v)\partial_2C(w,v)v(1-v) \right. \\
 & & +\partial_1C(u,v)\partial_1C(w,v)u(1-w)  \\
 & & -\partial_2C(u,v)C(w,v)(1-v) -\partial_2C(w,v)C(u,v)(1-v) \\
 & & -\partial_1C(u,v)[C(u,v)-C(w,v)u] -\partial_1C(w,v)C(u,v)(1-w) \\
 & & \left. +\partial_2C(u,v)\partial_1C(w,v)(C(w,v)-vw) +\partial_2C(w,v)\partial_1C(u,v)(C(u,v)-vu) \right\} dudw.
\end{array}$$
\end{theorem}

The proof of Theorem~\ref{thm:GaussChiV} is given in
Appendix~\ref{sec:proof_GaussChiV}.

To estimate $\chi^{Y\rightarrow X}$, one can select a fixed proportion $v$ of the data, with $v$ close to zero, and use the estimator $\widehat\chi^{Y\rightarrow X}_n(v)$. The accuracy of such an estimator relies on the selection criterion used. One can for instance select the value of $v$ that minimizes the asymptotic quadratic risk of the estimator $\widehat\chi^{Y\rightarrow X}_n(v)$. Thanks to Theorem~\ref{thm:GaussChiV}, along with the uniform integrability of the empirical copula process~\citep{GarcinNicolas}, we have
\begin{equation}\label{eq:riskQuad}\E\left[\left(\widehat\chi^{Y\rightarrow X}_n(v)-\chi^{Y\rightarrow X}\right)^2\right] = V_n(v) + \left[\int_0^1\left(\partial_2C(u,0) - \frac{C(u,v)}{v}\right)du\right]^2,
\end{equation}
with $\lim_{n\rightarrow\infty} nV_n(v)=\sigma^2_C(v)$. If the copula is two times differentiable, a Taylor expansion gives that the squared bias is bounded by $\int_0^1v^2\sup_{\xi\in[0,v]}\partial_2^2C(u,\xi)^2du/4$. Of course, the expression of the quadratic risk depends on $C$, which is unknown, but it is possible to express this risk explicitly for various parametric copulas and then to select the proportion $v$ minimizing the obtained quadratic risk, thus following a semi-parametric approach~\citep{GarcinNicolas}.

If $C$ is the independent copula, that is $C(u,v)=uv$, then we have:
$$\begin{array}{ccl}
\sigma^2_C(v) & = & \frac{2}{v^2} \int_0^1\int_0^w\left\{ uv-3uv^2w+uvw+uv^2 \right. \\
 & & -2uvw + 4uv^2w -2uv^2 \\
 & & \left. +2uv^2w -2uv^2w \right\} dudw \\
 & = & \frac{2}{v^2} \int_0^1\int_0^w\left\{ uv+uv^2w-uvw-uv^2 \right\} dudw \\
 & = & \frac{1}{v^2} \int_0^1\left\{ w^2v+v^2w^3-vw^3-w^2v^2 \right\} du \\
 & = & \frac{1}{v^2} \left\{ \frac{v}{3}+\frac{v^2}{4}-\frac{v}{4}-\frac{v^2}{3} \right\} \\
 & = & \frac{1}{12v}-\frac{1}{12}.
\end{array}$$
We remark that the variance is increasing when $v$ becomes smaller. When $v=1$, the variance is minimal and equal to zero. It is consistent with the intuition because, whatever the copula, $\chi^{Y\rightarrow X}(1)$, as well as $\widehat\chi_n^{Y\rightarrow X}(1)$, is trivially equal to 1/2, which is the conditional tail expectation of the independent copula. We also easily obtain that the bias appearing in equation~\eqref{eq:riskQuad} is equal to zero whatever $v$, so that the quadratic risk is limited to the above variance. For this copula, the best choice for $v$ in the nonparametric estimator would thus be $v=1$. 


For other copulas, obtaining $\sigma^2_C(v)$ is rarely straightforward  and a numerical calculation of this quantity is to be preferred. However, we can give another example for which we have an explicit formula of the quadratic risk, with the Farlie-Gumbel-Morgenstern (FGM) copula~\citep{Farlie,Gumbel60,Morgenstern}, which has a widespread use and many extensions~\citep{AmblardGirard,Hurlimann}. It is defined, for $\theta\in[-1,1]$, by $C(u,v)=uv+\theta u(1-u)v(1-v)$. Using SymPy, Python's library dedicated to symbolic computation, we get the following expression for $\sigma^2_C(v)$ in this FGM specification:
$$\sigma^2_C(v) = \frac{12 \theta^{2} v^{3} - 24 \theta^{2} v^{2} + 17 \theta^{2} v - 5 \theta^{2} - 15 v + 15}{180 v}.$$

In the case where $\theta=1$, we more simply have $\sigma^2_C(v) = (6v^3-12v^2+v+5)/90 v$. The bias for the FGM copula is
$$\int_0^1\partial_2C(u,0)du - \int_0^1\frac{C(u,v)}{v}du = \int_0^1 \left[u+\theta u(1-u)\right]du-\int_0^1 \left[u+\theta u(1-u)(1-v)\right]du =\frac{\theta v}{6}$$
and its conditional tail expectation is $1/2-\theta/6$, after equation~\eqref{eq:tcte_limit_copula}.

Figure~\ref{fig:FGM_Clayton} displays the asymptotic variance, squared bias, and quadratic risk of the estimator as functions of the threshold parameter $v$. Depending on the graph, data are generated either with the Clayton copula, $C(u,v) = (u^{-\theta} + v^{-\theta} - 1)^{-1/\theta}$ with $\theta=0.2$, or with the FGM copula with parameter $\theta = 1$. In the Clayton case, equation~\eqref{eq:tcte_limit_copula} yields $\chi^{Y \rightarrow X} = 0$, illustrating the lower-tail dependence structure of this copula. The asymptotic variance, squared bias, and quadratic risk defined in equation~\eqref{eq:riskQuad} are computed numerically and plotted as functions of $v$. The quadratic risk is minimized at $v = 0.030$ for $n = 100$ and at $v = 0.004$ for $n = 1{,}000$. These relatively small optimal thresholds reflect the stronger bias component of the estimator under the Clayton dependence structure. 
In the FGM case, for $n = 100$, the asymptotic quadratic risk is minimized at $v \approx 0.23$. When the sample size increases to $n = 1{,}000$, the asymptotic variance decreases substantially and the risk-minimizing threshold shifts to a smaller value, $v \approx 0.10$, illustrating the usual bias--variance trade-off. 

\begin{figure}[htb]
    \begin{subfigure}{\textwidth}
        \centering
        \includegraphics[width=\textwidth,
                    trim={0cm 1.5cm 0cm 0cm},clip]{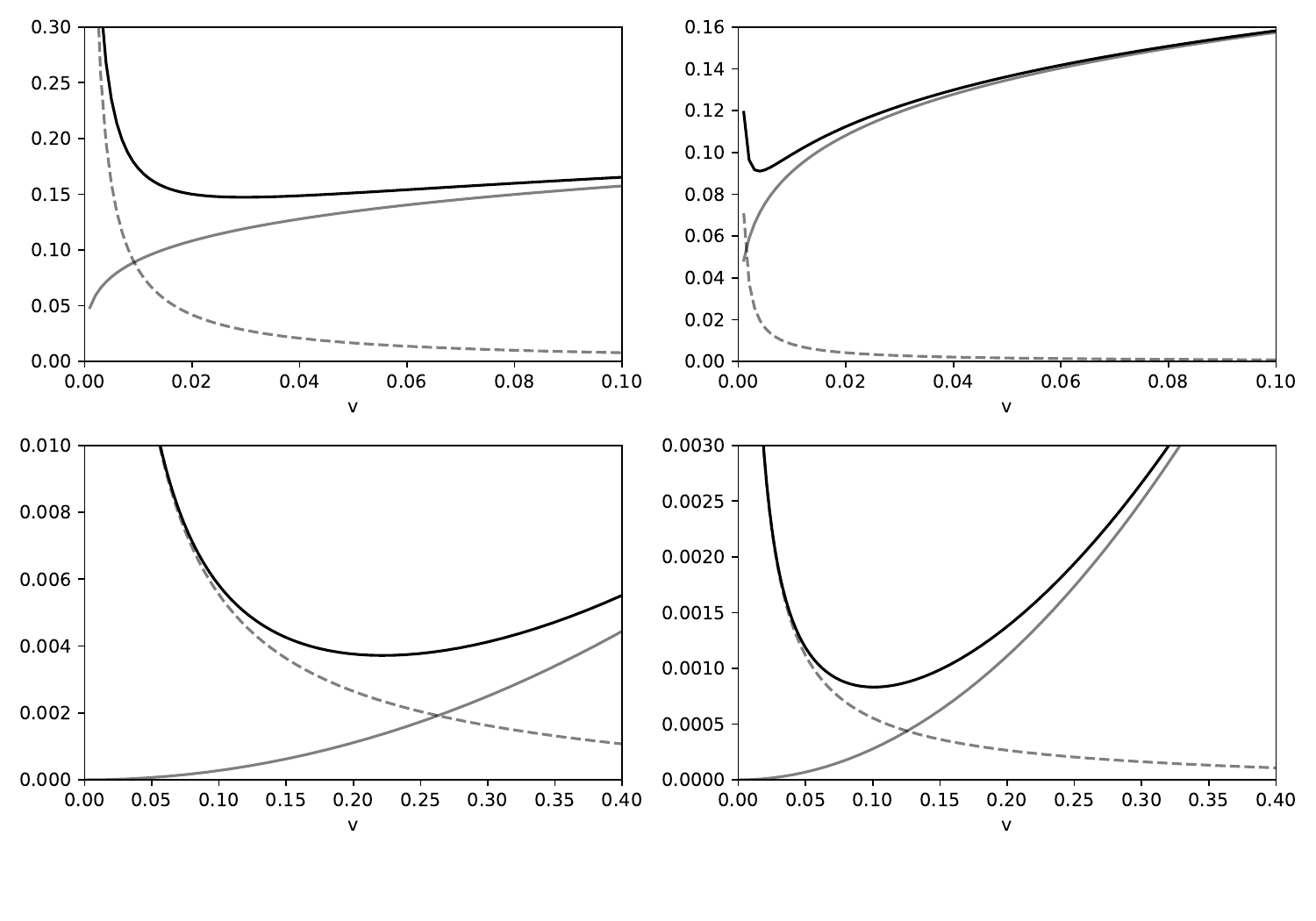} %
    \end{subfigure}
    \caption{Asymptotic variance $\sigma^2_C(v)/n$ (dashed curve), squared bias (grey curve), and quadratic risk (black curve) as functions of the threshold $v$. Top graphs correspond to the Clayton copula with $\theta = 0.2$ for $n = 100$ (left) and $n = 1{,}000$ (right). Bottom graphs correspond to the FGM copula with $\theta = 1$ for $n = 100$ (left) and $n = 1{,}000$ (right).}\label{fig:FGM_Clayton}
\end{figure}

Theorem~\ref{thm:GaussChiV} provides the asymptotic distribution of $\widehat\chi^{Y\rightarrow X}_n(v)$. In the case where the copula $C$ is exchangeable, the tail dependence is symmetric and $\widehat\chi^{X\rightarrow Y}_n(v)$, the empirical estimator of 
$$\chi^{X\rightarrow Y}(v)=\E[F_Y(Y)|F_X(X)\leq v]=1-\frac{1}{v}\int_0^1 C(v,u)du,$$ 
follows the same asymptotic distribution as $\widehat\chi^{Y\rightarrow X}_n(v)$. But without this exchangeability assumption, one has to consider the permuted copula $C_P$, which is defined by $C_P(u,v)=C(v,u)$. Then, starting from $\widehat\chi^{Y\rightarrow X}_n(v)$, one gets $\widehat\chi^{X\rightarrow Y}_n(v)$ simply by replacing $C$ by $C_P$. Therefore, the same conditions as the ones of Theorem~\ref{thm:Consistent} apply to have the consistency of $\widehat\chi^{X\rightarrow Y}_n(v)$ and $\widehat\chi^{X\rightarrow Y}_n(v_n)$. Regarding the asymptotic normality, Theorem~\ref{thm:GaussChiV} also holds for $\widehat\chi^{X\rightarrow Y}_n(v)$ if one replaces $\sigma^2_C(v)$ by
$$\begin{array}{ccl}
\sigma^2_{C_P}(v) & = & \frac{2}{v^2} \int_0^1\int_0^w\left\{ C_P(u,v)(1-C_P(w,v))+\partial_2C_P(u,v)\partial_2C_P(w,v)v(1-v) \right. \\
 & & +\partial_1C_P(u,v)\partial_1C_P(w,v)u(1-w)  \\
 & & -\partial_2C_P(u,v)C_P(w,v)(1-v) -\partial_2C_P(w,v)C_P(u,v)(1-v) \\
 & & -\partial_1C_P(u,v)[C_P(u,v)-C_P(w,v)u] -\partial_1C_P(w,v)C_P(u,v)(1-w) \\
 & & \left. +\partial_2C_P(u,v)\partial_1C_P(w,v)(C_P(w,v)-vw) +\partial_2C_P(w,v)\partial_1C_P(u,v)(C_P(u,v)-vu) \right\} dudw \\
 & = & \frac{2}{v^2} \int_0^1\int_0^w\left\{ C(v,u)(1-C(v,w))+\partial_1C(v,u)\partial_1C(v,w)v(1-v) \right. \\
 & & +\partial_2C(v,u)\partial_2C(v,w)u(1-w)  \\
 & & -\partial_1C(v,u)C(v,w)(1-v) -\partial_1C(v,w)C(v,u)(1-v) \\
 & & -\partial_2C(v,u)[C(v,u)-C(v,w)u] -\partial_2C(v,w)C(v,u)(1-w) \\
 & & \left. +\partial_1C(v,u)\partial_2C(v,w)(C(v,w)-vw) +\partial_1C(v,w)\partial_2C(v,u)(C(v,u)-vu) \right\} dudw.
\end{array}$$

Finally, from the estimation of $\chi^{Y\rightarrow X}$ and $\chi^{X\rightarrow Y}$, one can determine which variable, between $X$ and $Y$, has the greatest tail-focused impact on the distribution of the other variable. In a statistical perspective, we have first to determine whether $\widehat\chi_n^{Y\rightarrow X}(v)$ and $\widehat\chi_n^{X\rightarrow Y}(v)$ are significantly different from $1/2$, using Theorem~\ref{thm:GaussChiV}. Then, one has to study the statistical significance of $\widehat\chi_n(X,Y)(u,v):=\widehat\chi^{Y\rightarrow X}(v) - \widehat\chi^{X\rightarrow Y}(u)$. This estimator of the DTD measure $\chi(X,Y)(u,v):=\chi^{Y\rightarrow X}(v)-\chi^{X\rightarrow Y}(u)$ has the same properties of consistency as the ones of Theorem~\ref{thm:Consistent}. For the asymptotic normality, it is slightly more complex since the distributions of $\widehat\chi_n^{Y\rightarrow X}(v)$ and $\widehat\chi_n^{X\rightarrow Y}(u)$ are not independent of each other. Theorem~\ref{thm:GaussChiV_Diff}, proved in Appendix~\ref{sec:proof_GaussChiV_Diff}, gives more details about this. In particular, we can decompose the asymptotic variance in $\sigma^2_C(v)$, $\sigma^2_{C_P}(u)$, and a third part that results from the dependence between $\widehat\chi_n^{Y\rightarrow X}(v)$ and $\widehat\chi_n^{X\rightarrow Y}(u)$.

\begin{theorem}\label{thm:GaussChiV_Diff}
Let the bivariate copula $C$ have partial derivatives, with $\partial_1C$ continuous on $(0,1)\times[0,1]$ and $\partial_2C$ continuous on $[0,1]\times(0,1)$. Then, for all $u,v\in(0,1]^2$, 
$$\sqrt{n}\left(\widehat\chi_n(X,Y)(u,v)-\chi(X,Y)(u,v)\right) \overset{d}{\longrightarrow} \mathcal N\left(0,\sigma^2_C(v)+\sigma^2_{C_P}(u)+2\mathcal V_C(u,v)\right),$$
where 
$$\begin{array}{ccl}
\mathcal V_C(u,v) & = & - \frac{1}{uv}\int_0^u\int_0^v \left\{C(z,w)-C(u,w)C(z,v) + \partial_2C(u,w)\partial_2C(z,v)w(1-v)\right\}dwdz \\
 & & - \frac{1}{uv}\int_0^u\int_v^1 \left\{C(z,v)-C(u,w)C(z,v) + \partial_2C(u,w)\partial_2C(z,v)v(1-w)\right\}dwdz \\
 & & - \frac{1}{uv}\int_u^1\int_0^v \left\{C(u,w)(1-C(z,v)) + \partial_2C(u,w)\partial_2C(z,v)w(1-v)\right\}dwdz \\
 & & - \frac{1}{uv}\int_u^1\int_v^1 \left\{C(u,v)-C(u,w)C(z,v) + \partial_2C(u,w)\partial_2C(z,v)v(1-v)\right\}dwdz \\
 
 & & - \frac{1}{uv}\int_0^u\int_0^1 \partial_1C(u,w)\partial_1C(z,v)z(1-u)dwdz \\
  & & - \frac{1}{uv}\int_u^1\int_0^1 \partial_1C(u,w)\partial_1C(z,v)u(1-z)dwdz \\
  
 & & + \frac{1}{uv}\int_0^1\int_0^v \left\{\partial_2C(u,w)(C(z,w)-C(z,v)w) +\partial_2C(z,v)C(u,w)(1-v)\right\}dwdz \\
 & & + \frac{1}{uv}\int_0^1\int_v^1 \left\{\partial_2C(u,w)C(z,v)(1-w) +\partial_2C(z,v)(C(u,v)-C(u,w)v)\right\}dwdz \\
 
 & & + \frac{1}{uv}\int_0^u\int_0^1 \left\{\partial_1C(u,w)C(z,v)(1-u) +\partial_1C(z,v)(C(z,w)-C(u,w)z)\right\}dwdz  \\
  & & + \frac{1}{uv}\int_u^1\int_0^1 \left\{\partial_1C(u,w)(C(u,v)-C(z,v)u) +\partial_1C(z,v)C(u,w)(1-z)\right\}dwdz  \\
 
 & & - \frac{1}{uv}\int_0^1\int_0^1 \left\{\partial_2C(u,w)\partial_1C(z,v)(C(z,w)-zw) + \partial_2C(z,v)\partial_1C(u,w)(C(u,v)-uv)\right\}dwdz.
\end{array}$$
\end{theorem}

To calculate in practice the asymptotic variances of Theorems~\ref{thm:GaussChiV} and~\ref{thm:GaussChiV_Diff}, we need to estimate the partial derivatives of the copula function, $\partial_1 C$ and $\partial_2 C$. Since the true copula is unknown, these derivatives are approximated using finite-difference methods applied to the empirical copula $\widehat{C}_n$, which yields consistent estimates under mild smoothness conditions \citep{remillard2009testing}. In our empirical implementation, Sections~\ref{sec:sim-study} and~\ref{sec:data-application}, we choose the same bandwidth $h=n^{-1 / 2}$ as in \cite{remillard2009testing} and we replace each unknown partial derivative by the corresponding centered finite-difference of the empirical copula $\widehat{C}_n$. Concretely, we set
$$
\partial_1 \widehat{C}_n(u, v)= \begin{cases}\frac{\widehat{C}_n(2 h, v)}{2 h} & \text { if } u<h \\ \frac{\widehat{C}_n(u+h, v)-\widehat{C}_n(u-h, v)}{2 h} & \text { if } h \leq u \leq 1-h \\ \frac{\widehat{C}_n(1, v)-\widehat{C}_n(1-2 h, v)}{2 h} & \text { if } u>1-h\end{cases}
$$
and
$$
\partial_2 \widehat{C}_n(u, v)= \begin{cases}\frac{\widehat{C}_n(u, 2 h)}{2 h} & \text { if } v<h \\ \frac{\widehat{C}_n(u, v+h)-\widehat{C}_n(u, v-h)}{2 h} & \text { if } h \leq v \leq 1-h \\ \frac{\widehat{C}_n(u, 1)-\widehat{C}_n(u, 1-2 h)}{2 h} & \text { if } v>1-h.\end{cases}
$$

\section{Simulation Study} \label{sec:sim-study}

This section investigates the finite-sample behavior of the proposed directional tail dependence estimators through a simulation study. We consider two classes of dependence structures that are suited to assess directional effects in the tails: asymmetric copulas constructed via Khoudraji's device, and the skew-t copula.

\subsection{Khoudraji's Device}

Copula models can be made asymmetric using Khoudraji’s device \citep{khoudraji1997contributions, genest2012tests}. Let $C$ be an exchangeable copula. An asymmetric version of $C$ can be defined at all points $(u, v)$ in the unit square $[0,1]^2$ by the following expression: 
$$K_\delta(u, v) = u^\delta C\left(u^{1-\delta}, v\right),$$
where $\delta\in[0,1]$ is a parameter that introduces asymmetry into the copula. When $\delta=0$, we simply recover the copula $C$ and when $\delta=1$ it is the independent copula. We can easily generate random samples from $K_\delta$ \citep{khoudraji1997contributions, genest1998understanding}. If $(U_1,V)$ is sampled from the underlying copula $C$, with $U_1$ and $V$ uniform, and if $U_2$ is an independent uniform variable, then $K_\delta$ is the copula of the pair:
$$ \left( \max\left\{U_1^{1/(1-\delta)},U_2^{1/\delta} \right\}, V\right).$$
For the Khoudraji copula, the TCTEs admit the
expressions
$$\chi^{Y \rightarrow X}(v) =1-\frac{1}{v} \int_0^1 u^\delta C\left(u^{1-\delta}, v\right) d u,$$
and
$$\chi^{X \rightarrow Y}(v) =1-v^{\delta-1} \int_0^1 C\left(v^{1-\delta}, u\right) d u =\chi^{\widetilde X\rightarrow\widetilde Y}(v^{1-\delta}),$$
where $\widetilde X$ and $\widetilde Y$ are variables linked by the copula $C$ used in the definition of $K_{\delta}$. Considering the limit $v\rightarrow 0$, we get $\chi^{X \rightarrow Y}=\chi^{\widetilde X\rightarrow\widetilde Y}$ but the TCTE of $X$ given $Y$ may differ for the two copulas $K_{\delta}$ and $C$.

Figure~\ref{fig:simu-khoudraji} shows 5\,000 simulated observations produced by Khoudraji’s device, with a Gaussian, Clayton, or Frank copula. Within each panel the asymmetry parameter increases from left to right ($\delta = 0.1,\;0.25,\;0.5$); larger values of $\delta$ shift probability mass toward the upper–left corner of the unit square, illustrating the increasing directional asymmetry induced by the Khoudraji construction. 

\begin{figure}[htbp!]
    \begin{subfigure}{\textwidth}
        \centering
        \includegraphics[width=\textwidth,
                    trim={0cm 0.2cm 0cm 0cm},clip]{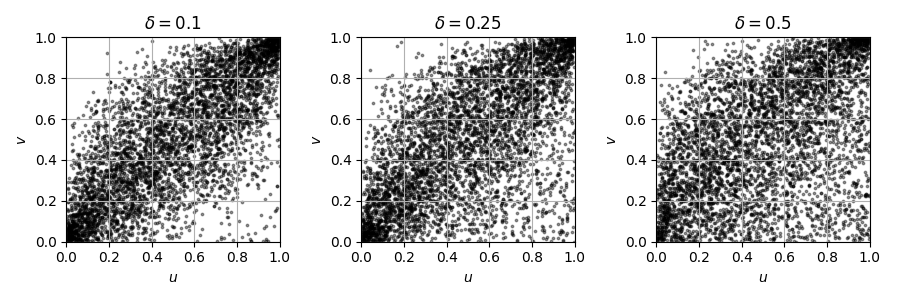} %
                    \caption{Gaussian Copula, $\rho=0.8$}
    \end{subfigure}
    \vspace{0.5cm}    
    \begin{subfigure}{\textwidth}
        \centering
        \includegraphics[width=\textwidth,
                    trim={0cm 0.2cm 0cm 0cm},clip]{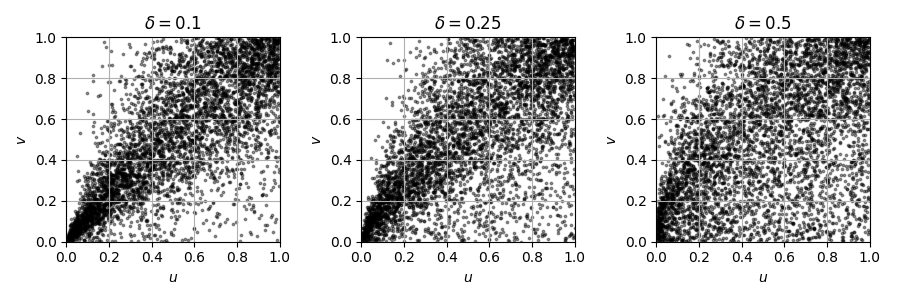} %
        \caption{ Clayton Copula, $\theta=3.5$ }
        \end{subfigure}
    \vspace{0.5cm}
    \begin{subfigure}{\textwidth}
        \centering
        \includegraphics[width=\textwidth,
                    trim={0cm 0.0cm 0cm 0cm},clip]{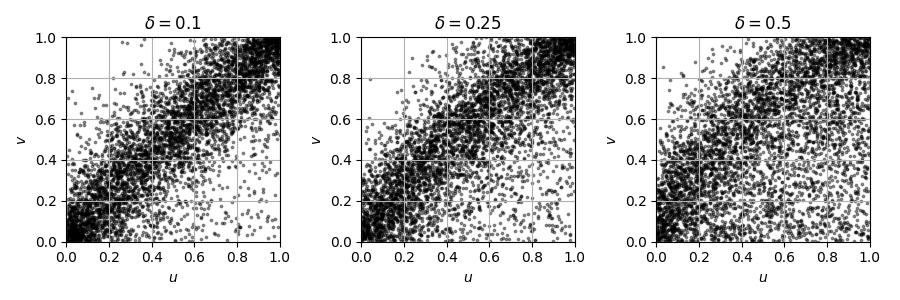} %
        \caption{Frank Copula, $\theta=10$ }
        \end{subfigure}
    \caption{Simulated samples from Khoudraji’s device, of size $n = 5{,}000$, for three types of copulas (one in each row) and three values of $\delta$ (one in each column).}\label{fig:simu-khoudraji}
\end{figure}

We perform simulations for a range of asymmetry parameters $\delta \in \{0.01, 0.10, 0.20, \ldots, 0.90, 0.99\}$ and for several families of exchangeable copulas. The corresponding parameter grids and copula specifications are summarized in Table~\ref{tab:simu-params-khoudraji}. The performance of the proposed estimators is evaluated using the root mean squared error (RMSE),
defined as
$$\text{RMSE}(\widehat{\chi}) = \sqrt{\frac{1}{N} \sum_{i=1}^{N} (\widehat{\chi}_i - \chi_i)^2},$$
where $N$ denotes the number of Monte Carlo replications.

\begin{table}[htbp!]
\caption{Parameter grids and CDF forms for the Khoudraji devices used in the simulations, with $\Phi$ and $\Phi_{\rho}$ the univariate and bivariate Gaussian CDFs.}
\label{tab:simu-params-khoudraji}
\centering
\resizebox{0.9\textwidth}{!}{%
\begin{tabular}{@{}llll@{}}
\toprule
Model & Notation & Parameter grid & Underlying copula $C(u,v)$ \\ \midrule
Gaussian & $K_\delta^{\text{Ga}}(u,v)$ & $\rho \in \{0.25, 0.50, 0.75\}$ &
$C^{\text{Ga}}_{\rho}(u,v)=\Phi_{\rho} \left(\Phi^{-1}(u),\Phi^{-1}(v)\right)$ \\
Clayton  & $K_\delta^{\text{Cl}}(u,v)$ & $\theta \in \{0.5, 1, 3\}$ &
$C^{\text{Cl}}_{\theta}(u,v)=\left(u^{-\theta}+v^{-\theta}-1\right)^{-1/\theta}$  \\
Frank    & $K_\delta^{\text{Fr}}(u,v)$ & $\theta \in \{2, 4, 10\}$ &
$C^{\text{Fr}}_{\theta}(u,v)=-\frac{1}{\theta}\,
\log \left[1+\tfrac{(e^{-\theta u}-1)(e^{-\theta v}-1)}{e^{-\theta}-1}\right]$ \\ \bottomrule
\end{tabular}}
\end{table}

Figure~\ref{fig:plot-gauss-khoudraji-chis} shows the theoretical TCTEs $\chi^{Y\rightarrow X}$, $\chi^{X\rightarrow Y}$, and the directional tail dependence $\chi(X,Y)$ for the Gaussian-Khoudraji case at four thresholds $v \in \{0.001, 0.01, 0.05, 0.1\}$, under either weak dependence ($\rho = 0.25$) or strong dependence ($\rho = 0.75$). When the asymmetry parameter $\delta$ is close to 1, both estimates of the TCTEs $\chi^{Y\rightarrow X}$ and $\chi^{X\rightarrow Y}$ are close to $1/2$ because the copula is close to the independent copula. When $\delta$ is close to 0, the copula is close to the exchangeable Gaussian copula, so that the two TCTEs $\chi^{Y\rightarrow X}$ and $\chi^{X\rightarrow Y}$ are close to each other. Therefore, there is no DTD for $\delta$ close to 0 or 1, but the DTD measure rises for intermediate values of $\delta$. Moreover, higher correlation levels shift all TCTE curves downward, indicating a greater overall tail dependence. Both the threshold choice and the strength of dependence have a substantial impact on the magnitude of the estimated coefficients. 
Table~\ref{tab:simu-gauss-khoudraji-chis} reports the bias and RMSE of $\widehat{\chi}_n^{Y\rightarrow X}$ and $\widehat{\chi}_n^{X\rightarrow Y}$ for the two thresholds $v = \sqrt{\log n / n}$ and $v = n^{-1/2}$. Bias is close to zero, and RMSE also decreases when $n$ increases. The threshold $v = \sqrt{\log n / n}$ yields slightly smaller errors than the threshold $v = n^{-1/2}$, however the difference is too small to conclude.

\begin{figure}[htbp!]
    \begin{subfigure}{\textwidth}
        \centering
        \includegraphics[width=\textwidth,
                    trim={0cm 0.8cm 0cm 0cm},clip]{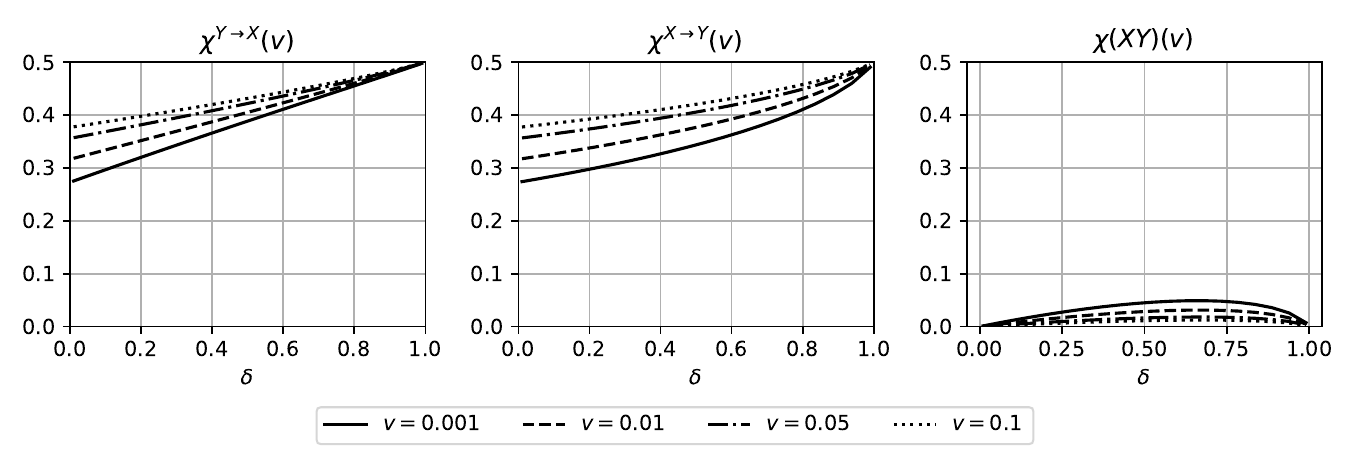} %
        \caption{ $\rho=0.25$}
    \end{subfigure}
    \vfill
    \begin{subfigure}{\textwidth}
        \centering
        \includegraphics[width=\textwidth,
                    trim={0cm 0cm 0cm 0cm},clip]{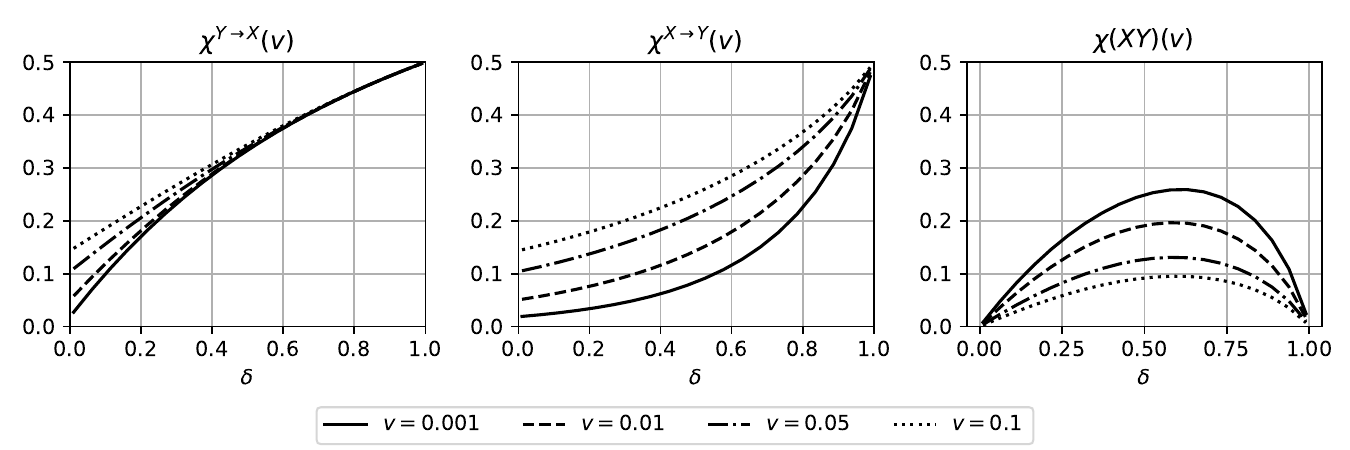} %
        \caption{ $\rho=0.75$} 
    \end{subfigure}
    \caption{Tail dependence measures (from left to right $\chi^{Y\rightarrow X}$, $\chi^{X\rightarrow Y}$, and $\chi(XY)(v)$) in the Gaussian-Khoudraji’s device, with parameter $\rho \in\{0.25,0.75\}$ for the Gaussian copula, for four values of $v$, as functions of $\delta$.}
    \label{fig:plot-gauss-khoudraji-chis} 
\end{figure}

\begin{table}[htbp!]
\setlength{\tabcolsep}{5pt} 
\caption{Bias and RMSE of the estimators $\widehat{\chi}_n^{Y\!\rightarrow\! X}(v)$ and $\widehat{\chi}_n^{X\!\rightarrow\! Y}(v)$ of $\chi^{Y\!\rightarrow\! X}(v)$ and $\chi^{X\!\rightarrow\! Y}(v)$, computed from 100 sample replications in the Gaussian-Khoudraji’s device. Results are shown for asymmetry parameters $\delta \in \{0.10, 0.25, 0.50, 0.75, 0.90\}$ and sample sizes  $n \in \{1{,}000; 5{,}000\}$. Panel~A uses the threshold $v = \sqrt{\log n\,/\,n}$ and panel~B uses $v = n^{-1/2}$. All simulations are performed using a Gaussian copula under three dependence scenarios: low ($\rho = 0.25$), moderate ($\rho = 0.50$), and strong ($\rho = 0.75$).}
\label{tab:simu-gauss-khoudraji-chis}
\resizebox{0.99\textwidth}{!}{%
\begin{tabular}{lllrrrrrrrrrr}
\toprule
\multicolumn{3}{c}{} & \multicolumn{10}{c}{$\delta$} \\
\multicolumn{3}{c}{} & \multicolumn{5}{c}{$n=1,000$} & \multicolumn{5}{c}{$n=5,000$} \\
\multicolumn{3}{c}{}   & 0.10 & 0.25 & 0.50 & 0.75 & 0.90  & \hspace{0.9cm}0.10 & 0.25 & 0.50 & 0.75 & 0.90 \\
\midrule
\multicolumn{13}{l}{\textbf{\textit{Panel A: $v = \sqrt{\log(n)/n} $}} }\\
$\rho=0.25$ & $\widehat{\chi}_n^{Y\rightarrow X}(v)$ & Bias & 0.001 & 0.003 & 0.001 & 0.002 & 0.000 & 0.002 & 0.002 & -0.004 & 0.000 & 0.002 \\
 &  & RMSE & 0.027 & 0.024 & 0.031 & 0.032 & 0.034 & 0.020 & 0.018 & 0.020 & 0.018 & 0.021 \\
 & $\widehat{\chi}_n^{X\rightarrow Y}(v)$ & Bias & 0.000 & 0.002 & 0.002 & 0.005 & 0.000 & 0.004 & 0.000 & -0.002 & 0.003 & -0.003 \\
 &  & RMSE & 0.025 & 0.027 & 0.031 & 0.030 & 0.031 & 0.017 & 0.019 & 0.020 & 0.020 & 0.020\vspace{0.25cm} \\
$\rho=0.5$ & $\widehat{\chi}_n^{Y\rightarrow X}(v)$ & Bias & 0.002 & 0.002 & 0.003 & -0.003 & -0.000 & 0.003 & 0.002 & 0.001 & 0.002 & -0.001 \\
 &  & RMSE & 0.022 & 0.028 & 0.029 & 0.031 & 0.031 & 0.015 & 0.018 & 0.022 & 0.019 & 0.020 \\
 & $\widehat{\chi}_n^{X\rightarrow Y}(v)$ & Bias & 0.003 & 0.001 & 0.001 & -0.002 & -0.003 & 0.004 & 0.002 & -0.000 & 0.003 & 0.001 \\
 &  & RMSE & 0.023 & 0.022 & 0.026 & 0.029 & 0.027 & 0.017 & 0.014 & 0.019 & 0.016 & 0.018\vspace{0.25cm} \\
$\rho=0.75$ & $\widehat{\chi}_n^{Y\rightarrow X}(v)$ & Bias & 0.002 & 0.000 & -0.002 & 0.003 & 0.004 & 0.001 & 0.003 & -0.005 & -0.001 & 0.004 \\
 &  & RMSE & 0.022 & 0.023 & 0.031 & 0.032 & 0.029 & 0.014 & 0.019 & 0.020 & 0.018 & 0.020 \\
 & $\widehat{\chi}_n^{X\rightarrow Y}(v)$ & Bias & 0.002 & 0.002 & -0.004 & 0.001 & 0.004 & 0.002 & 0.002 & -0.000 & 0.001 & 0.002 \\
 &  & RMSE & 0.016 & 0.017 & 0.021 & 0.025 & 0.028 & 0.009 & 0.010 & 0.011 & 0.015 & 0.016\vspace{0.25cm} \\
\multicolumn{13}{l}{\textbf{\textit{Panel B: $v = 1/ \sqrt{n} $}} } \\
$\rho=0.25$ & $\widehat{\chi}_n^{Y\rightarrow X}(v)$ & Bias & 0.004 & 0.011 & 0.014 & 0.011 & 0.011 & 0.010 & 0.006 & 0.006 & 0.000 & 0.005 \\
 &  & RMSE & 0.046 & 0.048 & 0.052 & 0.049 & 0.045 & 0.029 & 0.032 & 0.032 & 0.033 & 0.031 \\
 & $\widehat{\chi}_n^{X\rightarrow Y}(v)$ & Bias & 0.015 & 0.016 & 0.005 & 0.009 & 0.010 & 0.001 & 0.004 & 0.007 & 0.006 & 0.004 \\
 &  & RMSE & 0.042 & 0.044 & 0.044 & 0.050 & 0.053 & 0.032 & 0.031 & 0.030 & 0.033 & 0.029\vspace{0.25cm} \\
$\rho=0.5$ & $\widehat{\chi}_n^{Y\rightarrow X}(v)$ & Bias & 0.010 & 0.013 & 0.013 & 0.016 & 0.014 & 0.006 & 0.010 & 0.014 & 0.005 & 0.006 \\
 &  & RMSE & 0.041 & 0.046 & 0.046 & 0.046 & 0.050 & 0.029 & 0.035 & 0.032 & 0.031 & 0.036 \\
 & $\widehat{\chi}_n^{X\rightarrow Y}(v)$ & Bias & 0.016 & 0.017 & 0.018 & 0.015 & 0.019 & 0.005 & 0.005 & 0.008 & 0.008 & 0.006 \\
 &  & RMSE & 0.037 & 0.039 & 0.040 & 0.053 & 0.047 & 0.026 & 0.027 & 0.027 & 0.029 & 0.037\vspace{0.25cm} \\
$\rho=0.75$ & $\widehat{\chi}_n^{Y\rightarrow X}(v)$ & Bias & 0.015 & 0.008 & 0.014 & 0.014 & 0.015 & 0.010 & 0.003 & 0.009 & 0.009 & 0.001 \\
 &  & RMSE & 0.031 & 0.041 & 0.050 & 0.051 & 0.047 & 0.022 & 0.032 & 0.032 & 0.039 & 0.037 \\
 & $\widehat{\chi}_n^{X\rightarrow Y}(v)$ & Bias & 0.017 & 0.015 & 0.018 & 0.010 & 0.013 & 0.009 & 0.008 & 0.009 & 0.003 & 0.006 \\
 &  & RMSE & 0.021 & 0.021 & 0.031 & 0.034 & 0.047 & 0.010 & 0.013 & 0.017 & 0.026 & 0.027 \\
\bottomrule
\end{tabular}}
\end{table}

Figure~\ref{fig:plot-clayton-khoudraji-chis} repeats the analysis for a Clayton copula with dependence parameters $\theta = 0.5$ (low), $1$ (moderate), and $3$ (strong). The directional gap between $\chi^{Y\rightarrow X}$ and $\chi^{X\rightarrow Y}$ is larger than in the Gaussian case, possibly because the Clayton copula exhibits strong lower-tail dependence. Table~\ref{tab:simu-clayton-khoudraji-chis} presents the corresponding bias and RMSE results for the Clayton-Khoudraji model. As in the Gaussian case, the estimators exhibit relatively small bias and a clear decrease in RMSE as the sample size increases, confirming the good finite-sample performance of the proposed estimators under strong tail dependence.

\begin{figure}[htbp!]

    \begin{subfigure}{\textwidth}
        \centering
        \includegraphics[width=\textwidth,
                    trim={0cm 0.8cm 0cm 0cm},clip]{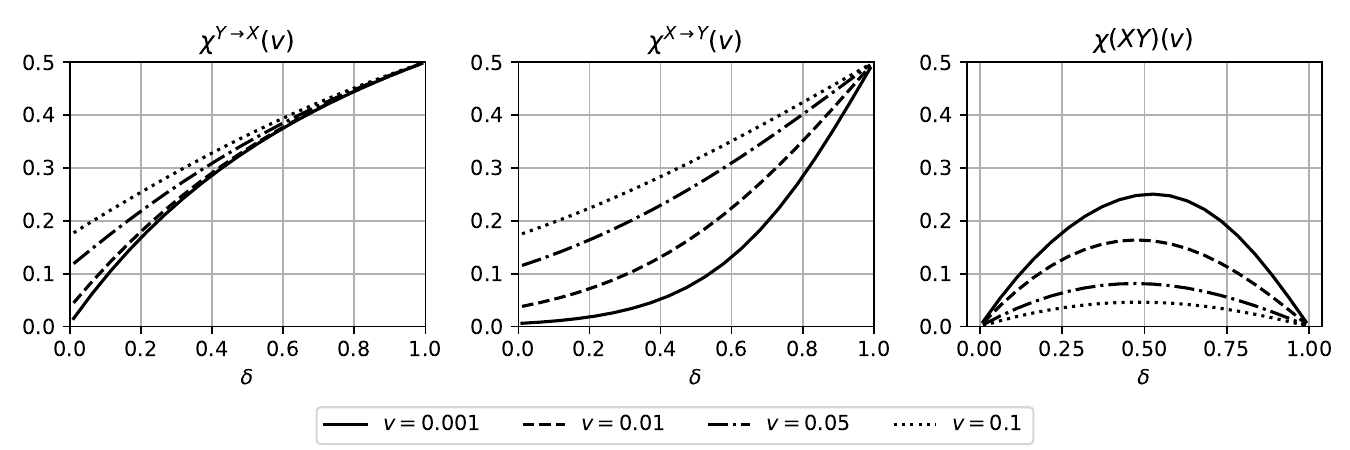} %
        \caption{ $\theta = 1$}
    \end{subfigure}
    \vfill
    \begin{subfigure}{\textwidth}
        \centering
        \includegraphics[width=\textwidth,
                    trim={0cm 0cm 0cm 0cm},clip]{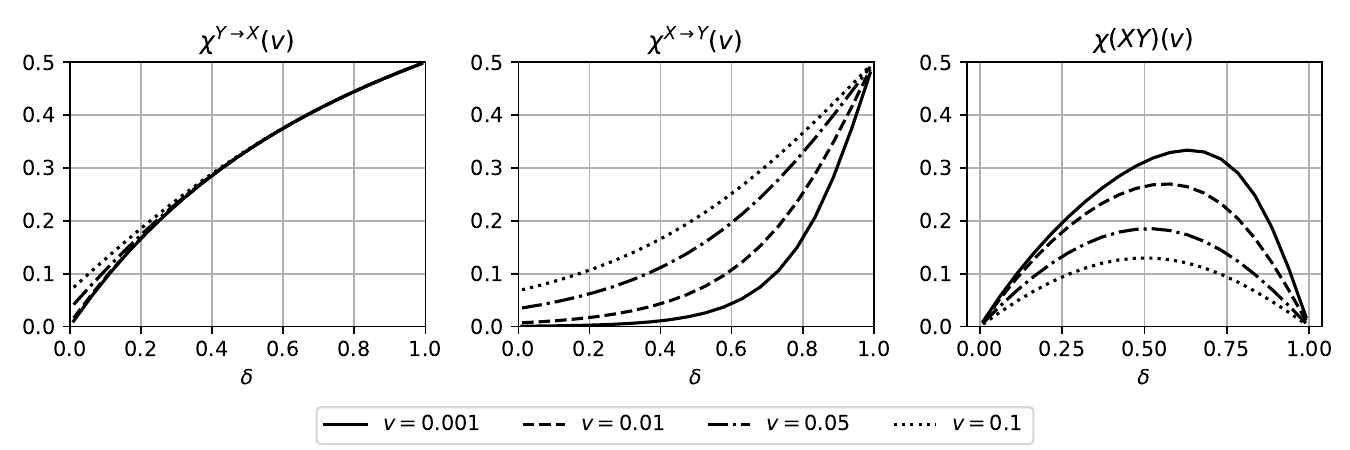} %
        \caption{ $\theta = 3$} 
    \end{subfigure}
        \caption{Tail dependence measures (from left to right $\chi^{Y\rightarrow X}$, $\chi^{X\rightarrow Y}$, and $\chi(XY)(v)$) in the Clayton-Khoudraji’s device, with parameter $\theta\in\{1,3\}$ for the Clayton copula, for four values of $v$, as functions of $\delta$.}\label{fig:plot-clayton-khoudraji-chis}
\end{figure}

\begin{table}[htbp!]
\setlength{\tabcolsep}{5pt} 
\caption{Bias and RMSE of the estimators $\widehat{\chi}_n^{Y\!\rightarrow\! X}(v)$ and $\widehat{\chi}_n^{X\!\rightarrow\! Y}(v)$ of $\chi^{Y\!\rightarrow\! X}(v)$ and $\chi^{X\!\rightarrow\! Y}(v)$, computed from 100 sample replications in the Clayton-Khoudraji’s device. Results are shown for asymmetry parameters $\delta \in \{0.10, 0.25, 0.50, 0.75, 0.90\}$ and sample sizes  $n \in \{1{,}000; 5{,}000\}$. Panel~A uses the threshold $v = \sqrt{\log n\,/\,n}$ and panel~B uses $v = n^{-1/2}$. All simulations are performed using a Clayton copula under three dependence scenarios: low ($\theta = 0.5$), moderate ($\theta = 1$), and strong ($\theta = 3$).}
    \label{tab:simu-clayton-khoudraji-chis}
\resizebox{0.99\textwidth}{!}{%
\begin{tabular}{lllrrrrrrrrrr}
\toprule
\multicolumn{3}{c}{} & \multicolumn{10}{c}{$\delta$} \\
\multicolumn{3}{c}{} & \multicolumn{5}{c}{$n=1,000$} & \multicolumn{5}{c}{$n=5,000$} \\
\multicolumn{3}{c}{}   & 0.10 & 0.25 & 0.50 & 0.75 & 0.90  & \hspace{0.9cm}0.10 & 0.25 & 0.50 & 0.75 & 0.90 \\
\midrule
\multicolumn{13}{l}{\textbf{\textit{Panel A: $v = \sqrt{\log(n)/n} $}} }\\
$\theta=0.5$ & $\widehat{\chi}_n^{Y\rightarrow X}(v)$ & Bias & -0.002 & -0.002 & 0.002 & -0.001 & -0.000 & 0.003 & -0.000 & 0.001 & -0.001 & -0.001 \\
 &  & RMSE & 0.027 & 0.031 & 0.032 & 0.028 & 0.031 & 0.019 & 0.022 & 0.020 & 0.020 & 0.020 \\
 & $\widehat{\chi}_n^{X\rightarrow Y}(v)$ & Bias & 0.000 & 0.002 & -0.000 & -0.000 & -0.002 & 0.006 & 0.000 & 0.000 & -0.003 & -0.000 \\
 &  & RMSE & 0.026 & 0.032 & 0.029 & 0.031 & 0.036 & 0.018 & 0.020 & 0.021 & 0.020 & 0.019\vspace{0.25cm} \\
$\theta=1$ & $\widehat{\chi}_n^{Y\rightarrow X}(v)$ & Bias & 0.002 & 0.004 & -0.000 & 0.003 & 0.004 & -0.000 & 0.001 & -0.004 & -0.001 & 0.003 \\
 &  & RMSE & 0.024 & 0.027 & 0.031 & 0.028 & 0.026 & 0.014 & 0.019 & 0.021 & 0.021 & 0.020 \\
 & $\widehat{\chi}_n^{X\rightarrow Y}(v)$ & Bias & -0.001 & 0.001 & 0.001 & 0.004 & 0.001 & 0.002 & -0.001 & -0.000 & -0.002 & 0.001 \\
 &  & RMSE & 0.021 & 0.024 & 0.026 & 0.029 & 0.032 & 0.013 & 0.013 & 0.018 & 0.020 & 0.020\vspace{0.25cm} \\
$\theta=3$ & $\widehat{\chi}_n^{Y\rightarrow X}v$ & Bias & 0.001 & -0.000 & -0.001 & -0.002 & -0.002 & 0.000 & 0.001 & 0.003 & 0.002 & -0.000 \\
 &  & RMSE & 0.022 & 0.025 & 0.027 & 0.032 & 0.034 & 0.013 & 0.018 & 0.017 & 0.019 & 0.019 \\
 & $\widehat{\chi}_n^{X\rightarrow Y}(v)$ & Bias & 0.001 & -0.000 & 0.001 & 0.002 & 0.000 & 0.001 & 0.002 & 0.001 & 0.001 & -0.000 \\
 &  & RMSE & 0.007 & 0.010 & 0.020 & 0.025 & 0.030 & 0.003 & 0.005 & 0.009 & 0.016 & 0.018 \\
\multicolumn{13}{l}{\textbf{\textit{Panel B: $v = 1/ \sqrt{n} $}} } \\
$\theta=0.5$ & $\widehat{\chi}_n^{Y\rightarrow X}(v)$ & Bias & 0.017 & 0.005 & 0.025 & 0.009 & 0.010 & 0.004 & -0.004 & 0.005 & -0.000 & -0.001 \\
 &  & RMSE & 0.045 & 0.050 & 0.052 & 0.052 & 0.049 & 0.027 & 0.033 & 0.033 & 0.035 & 0.035 \\
 & $\widehat{\chi}_n^{X\rightarrow Y}(v)$ & Bias & 0.016 & 0.020 & 0.007 & 0.008 & 0.015 & 0.013 & 0.006 & 0.010 & 0.009 & 0.006 \\
 &  & RMSE & 0.039 & 0.042 & 0.048 & 0.050 & 0.050 & 0.028 & 0.029 & 0.035 & 0.033 & 0.034\vspace{0.25cm} \\
$\theta=1$ & $\widehat{\chi}_n^{Y\rightarrow X}(v)$ & Bias & 0.013 & 0.028 & 0.012 & 0.007 & 0.007 & 0.009 & 0.008 & 0.007 & 0.008 & 0.009 \\
 &  & RMSE & 0.039 & 0.047 & 0.052 & 0.047 & 0.050 & 0.023 & 0.032 & 0.034 & 0.036 & 0.036 \\
 & $\widehat{\chi}_n^{X\rightarrow Y}(v)$ & Bias & 0.015 & 0.020 & 0.017 & 0.008 & 0.001 & 0.007 & 0.009 & 0.008 & 0.005 & 0.006 \\
 &  & RMSE & 0.028 & 0.032 & 0.049 & 0.047 & 0.046 & 0.015 & 0.018 & 0.025 & 0.031 & 0.033\vspace{0.25cm}\\
$\theta=3$ & $\widehat{\chi}_n^{Y\rightarrow X}v$ & Bias & 0.019 & 0.014 & 0.012 & 0.010 & 0.017 & 0.007 & 0.001 & 0.003 & 0.009 & 0.003 \\
 &  & RMSE & 0.036 & 0.049 & 0.050 & 0.050 & 0.048 & 0.022 & 0.030 & 0.034 & 0.035 & 0.035 \\
 & $\widehat{\chi}_n^{X\rightarrow Y}(v)$ & Bias & 0.019 & 0.018 & 0.020 & 0.013 & 0.016 & 0.009 & 0.010 & 0.010 & 0.010 & 0.003 \\
 &  & RMSE & 0.005 & 0.010 & 0.020 & 0.033 & 0.044 & 0.002 & 0.004 & 0.010 & 0.021 & 0.028 \\
\bottomrule
\end{tabular}}
\end{table}

Finally, Figure~\ref{fig:plot-frank-khoudraji-chis} examines the Frank copula for three levels of dependence: low ($\theta = 2$), moderate ($\theta = 4$), and strong ($\theta = 10$). Despite the absence of tail dependence in the Frank family, the directional tail coefficients respond sharply to changes in the asymmetry parameter $\delta$, demonstrating that the proposed measures remain sensitive to directional features driven purely by asymmetry. Table~\ref{tab:simu-frank-khoudraji-chis} reports the associated bias and RMSE values for $\widehat{\chi}_n^{Y\rightarrow X}$ and $\widehat{\chi}_n^{X\rightarrow Y}$, further confirming the accuracy and stability of the estimators across a wide range of dependence structures.

\begin{figure}[htbp!]
    \begin{subfigure}{\textwidth}
        \centering
        \includegraphics[width=\textwidth,
                    trim={0cm 0.8cm 0cm 0cm},clip]{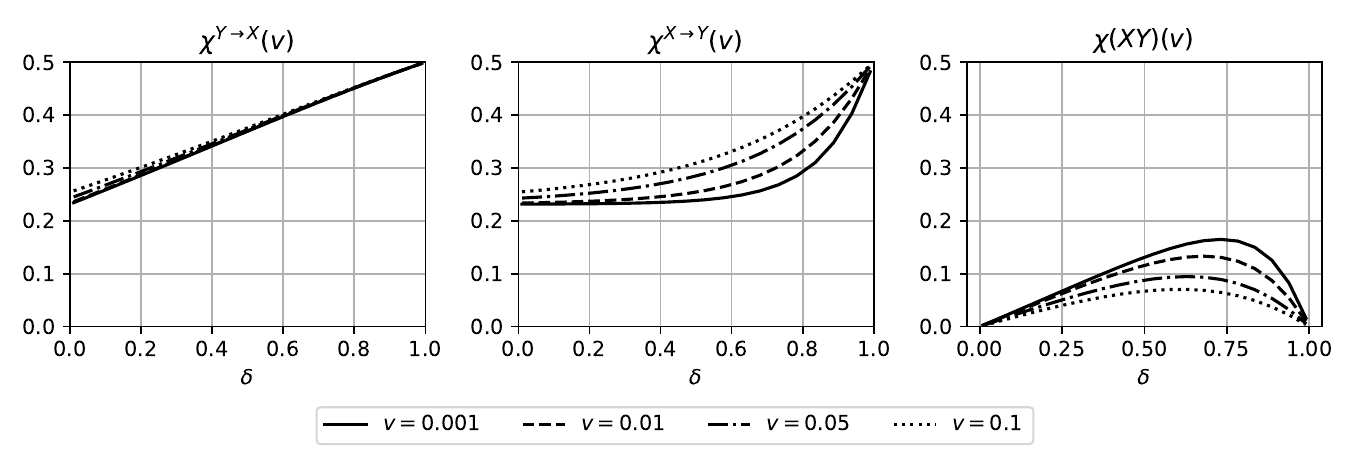} %
        \caption{ $\theta=4$} 
    \end{subfigure}
    \vfill
    \begin{subfigure}{\textwidth}
        \centering
        \includegraphics[width=\textwidth,
                    trim={0cm 0cm 0cm 0cm},clip]{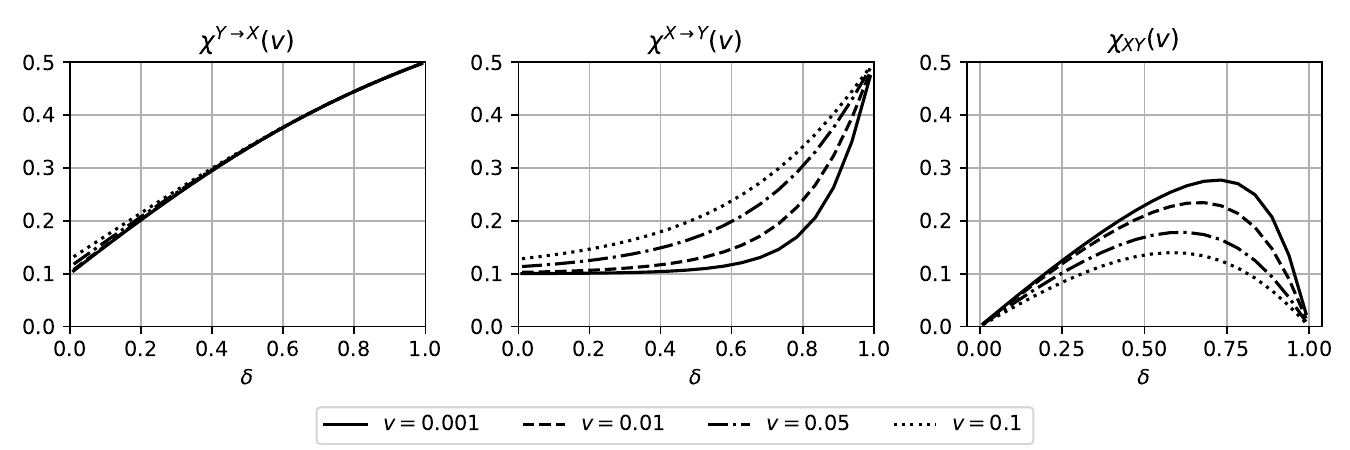} %
        \caption{ $\theta=10$}
    \end{subfigure}\caption{Tail dependence measures (from left to right $\chi^{Y\rightarrow X}$, $\chi^{X\rightarrow Y}$, and $\chi(XY)(v)$) in the Frank-Khoudraji’s device, with parameter $\theta\in\{4,10\}$ for the Frank copula, for four values of $v$, as functions of $\delta$.}\label{fig:plot-frank-khoudraji-chis}
\end{figure}

\begin{table}[htbp!]
\setlength{\tabcolsep}{5pt} 
\caption{Bias and RMSE of the estimators $\widehat{\chi}_n^{Y\!\rightarrow\! X}(v)$ and $\widehat{\chi}_n^{X\!\rightarrow\! Y}(v)$ of $\chi^{Y\!\rightarrow\! X}(v)$ and $\chi^{X\!\rightarrow\! Y}(v)$, computed from 100 sample replications in the Frank-Khoudraji’s device. Results are shown for asymmetry parameters $\delta \in \{0.10, 0.25, 0.50, 0.75, 0.90\}$ and sample sizes  $n \in \{1{,}000; 5{,}000\}$. Panel~A uses the threshold $v = \sqrt{\log n\,/\,n}$ and panel~B uses $v = n^{-1/2}$. All simulations are performed using a Clayton copula under three dependence scenarios: low ($\theta = 2$), moderate ($\theta = 4$), and strong ($\theta = 10$).}\label{tab:simu-frank-khoudraji-chis}
\resizebox{0.99\textwidth}{!}{%
\begin{tabular}{lllrrrrrrrrrr}
\toprule
\multicolumn{3}{c}{} & \multicolumn{10}{c}{$\delta$} \\
\multicolumn{3}{c}{} & \multicolumn{5}{c}{$n=1,000$} & \multicolumn{5}{c}{$n=5,000$} \\
\multicolumn{3}{c}{}   & 0.10 & 0.25 & 0.50 & 0.75 & 0.90  & \hspace{0.9cm}0.10 & 0.25 & 0.50 & 0.75 & 0.90 \\
\midrule
\multicolumn{13}{l}{\textbf{\textit{Panel A: $v = \sqrt{\log(n)/n} $}} }\\
$\theta=2$ & $\widehat{\chi}_n^{Y\rightarrow X}(v)$ & Bias & 0.002 & 0.004 & -0.003 & -0.001 & 0.004 & -0.003 & -0.002 & -0.001 & 0.001 & 0.002 \\
 &  & RMSE & 0.031 & 0.030 & 0.027 & 0.030 & 0.028 & 0.020 & 0.019 & 0.020 & 0.021 & 0.017 \\
 & $\widehat{\chi}_n^{X\rightarrow Y}(v)$ & Bias & -0.002 & 0.001 & 0.000 & -0.003 & 0.000 & 0.000 & 0.004 & 0.002 & 0.001 & -0.003 \\
 &  & RMSE & 0.032 & 0.031 & 0.029 & 0.027 & 0.031 & 0.018 & 0.016 & 0.019 & 0.018 & 0.021\vspace{0.25cm} \\
$\theta=4$ & $\widehat{\chi}_n^{Y\rightarrow X}(v)$ & Bias & 0.003 & 0.002 & 0.004 & 0.001 & 0.004 & 0.001 & 0.004 & 0.001 & -0.001 & 0.001 \\
 &  & RMSE & 0.026 & 0.027 & 0.028 & 0.033 & 0.028 & 0.015 & 0.018 & 0.020 & 0.020 & 0.019 \\
 & $\widehat{\chi}_n^{X\rightarrow Y}(v)$ & Bias & 0.004 & 0.001 & 0.001 & -0.001 & 0.001 & 0.004 & 0.003 & 0.003 & -0.000 & -0.002 \\
 &  & RMSE & 0.022 & 0.024 & 0.025 & 0.026 & 0.031 & 0.015 & 0.015 & 0.016 & 0.018 & 0.019\vspace{0.25cm} \\
$\theta=10$ & $\widehat{\chi}_n^{Y\rightarrow X}(v)$ & Bias & -0.003 & -0.001 & 0.002 & -0.000 & -0.006 & 0.001 & -0.000 & 0.002 & -0.000 & 0.003 \\
 &  & RMSE & 0.018 & 0.026 & 0.028 & 0.031 & 0.027 & 0.012 & 0.016 & 0.019 & 0.019 & 0.020 \\
 & $\widehat{\chi}_n^{X\rightarrow Y}(v)$ & Bias & 0.000 & 0.001 & 0.002 & 0.001 & -0.003 & 0.001 & 0.003 & 0.004 & 0.002 & 0.000 \\
 &  & RMSE & 0.010 & 0.013 & 0.016 & 0.021 & 0.028 & 0.006 & 0.007 & 0.009 & 0.013 & 0.017 \\

\multicolumn{13}{l}{\textbf{\textit{Panel B: $v = 1/ \sqrt{n} $}} } \\

 $\theta=2$ & $\widehat{\chi}_n^{Y\rightarrow X}(v)$ & Bias & 0.014 & 0.013 & 0.005 & 0.011 & 0.010 & 0.003 & 0.001 & 0.005 & 0.003 & 0.007 \\
 &  & RMSE & 0.043 & 0.046 & 0.051 & 0.055 & 0.046 & 0.033 & 0.036 & 0.035 & 0.039 & 0.030 \\
 & $\widehat{\chi}_n^{X\rightarrow Y}(v)$ & Bias & 0.011 & 0.009 & 0.018 & 0.015 & 0.014 & 0.005 & 0.011 & 0.008 & 0.005 & 0.007 \\
 &  & RMSE & 0.042 & 0.045 & 0.047 & 0.045 & 0.047 & 0.029 & 0.028 & 0.032 & 0.033 & 0.033\vspace{0.25cm} \\
$\theta=4$ & $\widehat{\chi}_n^{Y\rightarrow X}(v)$ & Bias & 0.012 & 0.012 & 0.018 & 0.008 & 0.007 & 0.005 & 0.013 & 0.009 & 0.008 & 0.005 \\
 &  & RMSE & 0.041 & 0.050 & 0.054 & 0.049 & 0.047 & 0.028 & 0.033 & 0.031 & 0.034 & 0.036 \\
 & $\widehat{\chi}_n^{X\rightarrow Y}(v)$ & Bias & 0.011 & 0.012 & 0.014 & 0.012 & 0.014 & 0.015 & 0.010 & 0.010 & 0.008 & 0.009 \\
 &  & RMSE & 0.035 & 0.036 & 0.038 & 0.041 & 0.046 & 0.028 & 0.025 & 0.024 & 0.026 & 0.029\vspace{0.25cm} \\
$\theta=10$ & $\widehat{\chi}_n^{Y\rightarrow X}(v)$ & Bias & 0.018 & 0.016 & 0.014 & 0.021 & 0.015 & 0.005 & 0.007 & 0.010 & 0.007 & 0.010 \\
 &  & RMSE & 0.032 & 0.040 & 0.048 & 0.053 & 0.046 & 0.022 & 0.031 & 0.029 & 0.032 & 0.035 \\
 & $\widehat{\chi}_n^{X\rightarrow Y}(v)$ & Bias & 0.016 & 0.020 & 0.014 & 0.013 & 0.005 & 0.009 & 0.009 & 0.007 & 0.008 & 0.004 \\
 &  & RMSE & 0.020 & 0.020 & 0.021 & 0.030 & 0.042 & 0.012 & 0.014 & 0.012 & 0.019 & 0.027 \\

\bottomrule
\end{tabular}}
\end{table}

\subsection{Skew-t Copula}

To further assess the performance of the proposed directional tail dependence test, we conduct a simulation study based on the skew-$t$ copula, which allows for asymmetric tail dependence driven by skewness parameters $\alpha_1$ and $\alpha_2$ associated with each margin \citep{kollo2010parameter}.

An algorithm to simulate random samples from the skew-t copula is given in \cite{kollo2010parameter}. In our setup, we consider increasing degrees of asymmetry by setting $\alpha_1 = -\alpha_2$ and varying the magnitude of $\alpha_1$ from 0.1 to 2.0, while fixing the number of degrees of freedom at $\nu = 3$. Two correlation levels are examined, $\rho = 0.25$ and $\rho = 0.5$, to represent weak and moderate dependence, respectively. Figure~\ref{fig:simu-skewt} displays simulations.

\begin{figure}[htbp!]
    \begin{subfigure}{\textwidth}
        \centering
        \includegraphics[width=\textwidth,
                    trim={0cm 0.4cm 0cm 0cm},clip]{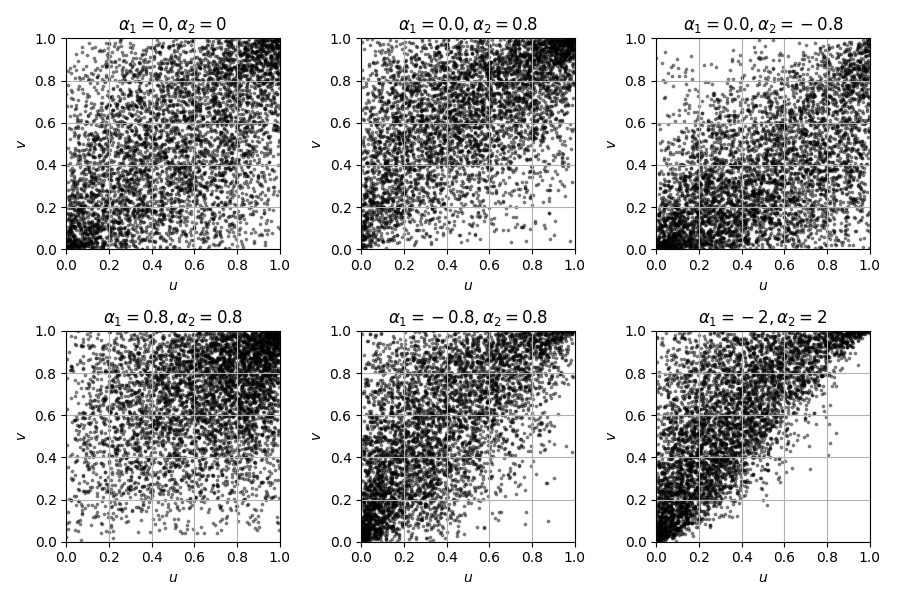} %
    \end{subfigure}
    \caption{Simulated samples from the skew-$t$ copula, of size $n = 5{,}000$, for different combinations of the skewness parameters $\alpha_1$ and $\alpha_2$, a number of degrees of freedom $\nu = 10$, and a correlation coefficient $\rho = 0.5$.}\label{fig:simu-skewt}
\end{figure}

For each configuration, we generate 100 bivariate samples of size $n = 1,000$, and compute the directional tail coefficients $\widehat{\chi}_n^{Y \rightarrow X}(v)$ and $\widehat{\chi}_n^{X \rightarrow Y}(v)$ as well as their difference $\widehat{\chi}_n(X,Y)(v)$. The empirical power of the test is then evaluated at the 5\% significance level under the null hypothesis of no directional dependence ($\chi(X,Y) = 0$), using two thresholds: $v = \sqrt{\log(n)/n}$ and $v = n^{-1/2}$.

The results, presented in Table~\ref{tab:empirical-power}, indicate a clear increase in $\widehat{\chi}_n(X,Y)(v)$ with growing skewness, confirming that the directional signal strengthens as asymmetry becomes more pronounced. Correspondingly, the power of the test improves substantially, achieving over 90\% power under moderate dependence even when $|\alpha_1|$ is close to 1, particularly when using the threshold $v = \sqrt{\log(n)/n}$, which consistently outperforms the simpler choice $v = n^{-1/2}$.

\begin{table}[htbp!]
\setlength{\tabcolsep}{9pt} 
\caption{Estimates and empirical power of the test for detecting directional tail dependence ($\chi(X,Y) \neq 0$) under skew-t copula with $\nu = 3$. The skewness parameters are set such that $\alpha_1 = -\alpha_2$, introducing increasing asymmetry. Average estimats $\widehat{\chi}_n^{Y \rightarrow X}(v)$, $\widehat{\chi}_n^{X \rightarrow Y}(v)$, and their difference $\widehat{\chi}_n(X,Y)(v)$ are presented. Power is computed as the percentage of rejections at the 5\% level over 100 replications. Results are shown for two correlation settings ($\rho = 0.25$ and $\rho = 0.5$) and two threshold choices ($v = \sqrt{\log(n)/n}$ and $v = n^{-1/2}$).}
\label{tab:empirical-power}
\begin{tabular}{llrrrrrrrr}
\toprule
 &  & \multicolumn{8}{c}{$\alpha_1 = -\alpha_2$}\\
 &  & 0.10 & 0.37 & 0.64 & 0.91 & 1.19 & 1.46 & 1.73 & 2.00 \\
\midrule
\multicolumn{10}{l}{\textbf{\textit{Panel A: $v = \sqrt{\log(n)/n} $}} }\\
 $\rho=0.25$ & $\widehat{\chi}_n^{Y\rightarrow X}(v)$ & 0.39 & 0.38 & 0.37 & 0.36 & 0.34 & 0.34 & 0.33 & 0.33 \\
 & $\widehat{\chi}_n^{X\rightarrow Y}(v)$ & 0.37 & 0.32 & 0.27 & 0.23 & 0.20 & 0.18 & 0.17 & 0.15 \\
 & $\widehat{\chi}_n(X,Y)(v)$ & 0.01 & 0.06 & 0.10 & 0.13 & 0.14 & 0.16 & 0.17 & 0.18 \\
 & Power (\%) & 0.00 & 10.00 & 46.00 & 83.00 & 95.00 & 100.00 & 100.00 & 99.00\vspace{0.25cm} \\
$\rho=0.5$ & $\widehat{\chi}_n^{Y\rightarrow X}(v)$ & 0.26 & 0.27 & 0.26 & 0.25 & 0.24 & 0.24 & 0.23 & 0.23 \\
 & $\widehat{\chi}_n^{X\rightarrow Y}(v)$ & 0.25 & 0.22 & 0.19 & 0.16 & 0.14 & 0.12 & 0.11 & 0.11 \\
 & $\widehat{\chi}_n(X,Y)(v)$ & 0.02 & 0.05 & 0.07 & 0.09 & 0.10 & 0.11 & 0.12 & 0.12 \\
 & Power (\%) & 0.00 & 13.00 & 51.00 & 86.00 & 94.00 & 100.00 & 99.00 & 98.00\vspace{0.25cm} \\
\multicolumn{9}{l}{\textbf{\textit{Panel B: $v = 1/ \sqrt{n} $}} } \\
$\rho=0.25$ & $\widehat{\chi}_n^{Y\rightarrow X}(v)$ & 0.39 & 0.39 & 0.37 & 0.36 & 0.36 & 0.35 & 0.35 & 0.33 \\
 & $\widehat{\chi}_n^{X\rightarrow Y}(v)$ & 0.35 & 0.29 & 0.23 & 0.18 & 0.15 & 0.12 & 0.11 & 0.10 \\
 & $\widehat{\chi}_n(X,Y)(v)$ & 0.04 & 0.10 & 0.13 & 0.18 & 0.21 & 0.22 & 0.24 & 0.23 \\
 & Power (\%) & 4.00 & 18.00 & 31.00 & 55.00 & 81.00 & 83.00 & 93.00 & 97.00\vspace{0.25cm} \\
$\rho=0.5$ & $\widehat{\chi}_n^{Y\rightarrow X}(v)$ & 0.24 & 0.25 & 0.25 & 0.25 & 0.24 & 0.23 & 0.23 & 0.23 \\
 & $\widehat{\chi}_n^{X\rightarrow Y}(v)$ & 0.23 & 0.18 & 0.15 & 0.12 & 0.10 & 0.08 & 0.08 & 0.07 \\
 & $\widehat{\chi}_n(X,Y)(v)$ & 0.01 & 0.07 & 0.10 & 0.13 & 0.14 & 0.14 & 0.16 & 0.16 \\
 & Power (\%) & 2.00 & 9.00 & 25.00 & 44.00 & 63.00 & 70.00 & 84.00 & 85.00 \\
\bottomrule
\end{tabular}%
\end{table}

\section{Ocean Data Application} \label{sec:data-application}

Oceanographic data have recently been shown to exhibit asymmetric dependence structures \citep{vanem2016joint, zhang2018modeling, lyu2023testing}. In this work, we analyze data collected from the Western Gulf of Alaska buoy (Station 46001)\footnote{\url{https://www.ndbc.noaa.gov/station_page.php?station=46001}}, maintained by the U.S. National Data Buoy Center (NDBC). This station provides extensive longitudinal datasets on meteorological and spectral wave conditions dating back to 1972, which we use to investigate the presence of directional dependence in extreme events. Our dataset comprises 509,379 hourly observations over 50 years, from January 1976 to January 2026. We specifically focus on measurements of significant wave height (WVHT), average wave period (APD), wind speed (WSPD), dominant wave period (DPD), and sea level pressure (PRES).
We observe strong autocorrelation in the oceanographic data. Following \cite{vanem2016joint}, a two-stage pre-processing methodology was employed to mitigate serial dependence and seasonal non-stationarity. First, short-term autocorrelation due to persistence in the wave generation process was addressed through a systematic subsampling procedure, which reduced the hourly time series to weekly intervals. 

This resulted in 1,937 weekly observations for each variable. Second, seasonality was eliminated by calculating the seasonal mean ($\mu_j$) and standard deviation ($\sigma_j$) for each week $j \in \{1, \dots, 52\}$ of the annual cycle. Specifically, original observations ($X_i$) were transformed into pre-processed values ($Y_i$) using the relation:
\begin{equation}
    Y_i = \frac{X_i - \mu_{w(i)}}{\sigma_{w(i)}} + M,
\end{equation}
where $w(i)$ is a function that maps the observation index $i$ to its corresponding week of the year, and $M$ represents the overall mean of the dataset. Table~\ref{tab:descriptive_stats} presents a detailed overview and descriptive statistics of the transformed variables.


\begin{table}[htbp!]
\centering
\caption{Descriptive statistics of wave and wind variables.}
\label{tab:descriptive_stats}
\resizebox{0.8\textwidth}{!}{%
\begin{tabular}{l p{7cm} r r r r}
\toprule
 Variable & \multicolumn{1}{c}{Description} & Min & Median & Q99 & Max \\
\midrule
WVHT & Significant wave height (m), average height of the highest one-third waves & 0.40 & 2.40 & 7.20 & 9.32 \\
APD & Average wave period (s), mean time between successive crests & 3.45 & 6.51 & 9.90 & 11.60 \\
DPD & Dominant wave period (s), period with greatest spectral energy & 3.57 & 10.00 & 16.70 & 20.00 \\
WSPD & Wind speed ($\text{m}.\text{s}^{-1}$), hourly-mean wind speed at anemometer height & 0.00 & 7.30 & 17.05 & 21.10 \\
PRES & Abnormal sea level pressure (mb), the abnormal atmospheric pressure adjusted to mean sea level & -31.55 & 5.05 & 43.36 & 58.05 \\
\bottomrule
\end{tabular}%
}
\end{table}

Figure~\ref{fig:ocean_data} provides a visual summary of the dependence structure. The copula plots, in particular, highlight the presence of strong asymmetries in the joint tails, indicating potential directional dependence in the extremes of the oceanographic variables.

\begin{figure}[htbp!]
        \centering
        \includegraphics[width=\textwidth,
                    trim={2cm 2cm 2cm 2cm},clip]{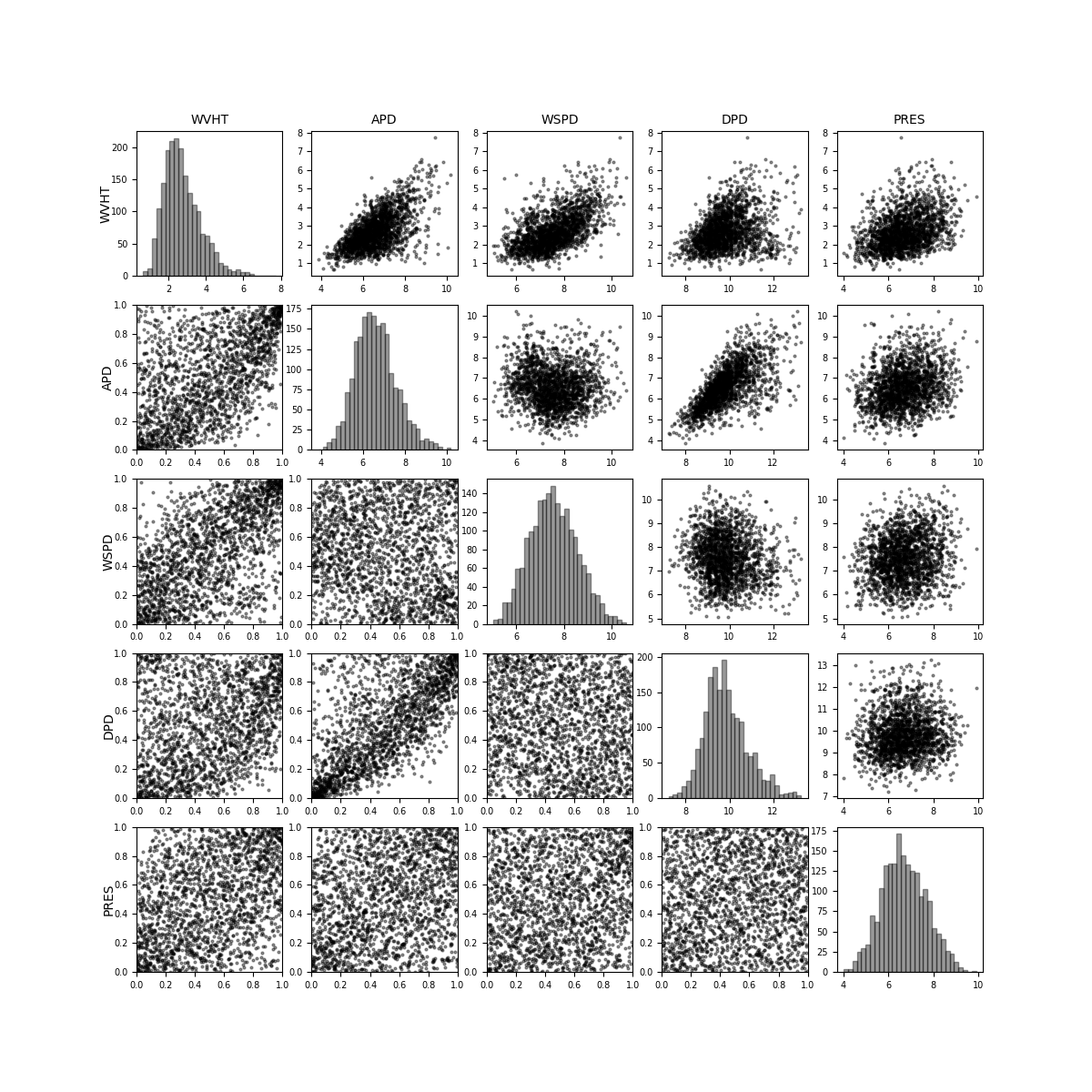} %
    \caption{Bivariate behaviour of wind and wave variables from the National Data Buoy Center (station 46001). The diagonal panels represent histograms of each variable’s marginal distribution, the upper-triangle plots show the raw bivariate data, and the lower-triangle plots show the same pairs after transforming each series to scaled ranks, showing their empirical copulas. The five key variables analysed are WVHT, APD, WSPD, DPD, and PRES.}\label{fig:ocean_data}
\end{figure}

In this real application, we are mostly interested in large values of the variables, more than on extremely small ones. As explained at the end of Section~\ref{sec:dir-tail-dep}, the above methological framework can easily be extended to upper-tail dependence. We thus apply a simple transformation of all our variables, by multiplying them by -1, so that the standard lower-tail framework can be applied to these transformed variables. Table~\ref{tab:direction_inference_metrics} reports the estimated directional tail dependence coefficients $\widehat{\chi}_n(X,Y)(v)$ for all pairs, along with $p$-values from the DTD test. Several pairs, in particular APD with each of the other variables, exhibit values of $\widehat{\chi}_n(X,Y)(v)$ that are statistically significantly different from $0$. This often indicates that extremes at one variable tend to condition those at another, rather than occurring symmetrically.

\begin{table}[htbp!]
\setlength{\tabcolsep}{8pt} 
\caption{Summary of directional inference metrics between pairs of variables: $\widehat{\chi}_n^{Y\rightarrow X}(v)$, $\widehat{\chi}_n^{X\rightarrow Y}(v)$, $\widehat{\chi}_n(X,Y)(v)$, as well as the estimates $\widehat{\sigma}^2_{C,n}(v)$, $\widehat{\sigma}^2_{C_P,n}(v)$, and $\widehat{\mathcal{V}}_{C,n}(v,v)$, of the variances $\sigma^2_C(v)$ and $\sigma^2_{C_P}(v)$ and covariance $\mathcal{V}_C(v,v)$, based on empirical copulas. The last column reports the p-value of the statistical test assessing the null hypothesis of no DTD ($\chi(X,Y) \neq 0$). Results are presented for two threshold levels ($v = \sqrt{\log(n)/n}$ and $v = n^{-1/2}$).} \label{tab:direction_inference_metrics}
\resizebox{\textwidth}{!}{%
\begin{tabular}{llrrrrrrr}
\toprule
X & Y & $\widehat{\chi}_n^{Y\rightarrow X}(v)$ & $\widehat{\chi}_n^{X\rightarrow Y}(v)$ & $\widehat{\chi}_n(X,Y)(v)$ & $\widehat{\sigma}^2_{C,n}(v)$ & $\widehat{\sigma}^2_{C_P,n}(v)$ & $\widehat{\mathcal{V}}_{C,n}(v,v)$ & p-value \\
\midrule
\multicolumn{9}{l}{\textbf{\textit{Panel A: $v = \sqrt{\log(n)/n} $}} }\\
WVHT & APD & 0.23 & 0.11 & 0.11 & 1.46 & 0.44 & -0.31 & 0.00 \\
WVHT & WSPD & 0.16 & 0.21 & -0.06 & 0.64 & 1.18 & -0.22 & 0.04 \\
WVHT & DPD & 0.51 & 0.25 & 0.26 & 2.36 & 0.90 & 0.16 & 0.00 \\
WVHT & PRES & 0.29 & 0.32 & -0.02 & 1.22 & 1.31 & 0.09 & 0.56 \\
APD & WSPD & 0.42 & 0.52 & -0.11 & 1.63 & 2.16 & 0.30 & 0.02 \\
APD & DPD & 0.27 & 0.13 & 0.13 & 1.41 & 0.38 & -0.18 & 0.00 \\
APD & PRES & 0.37 & 0.37 & -0.00 & 1.55 & 1.55 & 0.13 & 1.00 \\
WSPD & DPD & 0.60 & 0.52 & 0.08 & 1.66 & 1.71 & 0.10 & 0.07 \\
WSPD & PRES & 0.42 & 0.34 & 0.09 & 1.90 & 1.48 & 0.15 & 0.05 \\
DPD & PRES & 0.48 & 0.50 & -0.03 & 1.58 & 1.68 & 0.10 & 0.53 \\
\multicolumn{9}{l}{\textbf{\textit{Panel B: $v = 1/ \sqrt{n} $}} } \\
WVHT & APD & 0.21 & 0.06 & 0.16 & 4.41 & 0.43 & -0.19 & 0.00 \\
WVHT & WSPD & 0.11 & 0.19 & -0.07 & 0.85 & 3.19 & -0.04 & 0.10 \\
WVHT & DPD & 0.48 & 0.19 & 0.29 & 6.22 & 1.63 & 0.59 & 0.00 \\
WVHT & PRES & 0.28 & 0.29 & -0.01 & 3.18 & 2.37 & 0.20 & 0.82 \\
APD & WSPD & 0.40 & 0.52 & -0.12 & 4.08 & 5.52 & 0.74 & 0.12 \\
APD & DPD & 0.24 & 0.10 & 0.14 & 3.25 & 0.60 & 0.08 & 0.00 \\
APD & PRES & 0.33 & 0.40 & -0.07 & 3.99 & 4.04 & 0.43 & 0.30 \\
WSPD & DPD & 0.60 & 0.53 & 0.07 & 3.96 & 3.86 & -0.02 & 0.26 \\
WSPD & PRES & 0.43 & 0.31 & 0.12 & 4.80 & 3.25 & 0.22 & 0.06 \\
DPD & PRES & 0.44 & 0.51 & -0.06 & 3.56 & 4.33 & -0.01 & 0.31 \\
\bottomrule
\end{tabular}%
}
\end{table}

These results are visually summarized in Figure~\ref{fig:ocean-network}, which presents a  directed graph where arrows represent the DTD links identified as statistically  significant at the 5\% level. The figure displays links that are significant for both  thresholds ($v = \sqrt{\log(n)/n}$ and $v = n^{-1/2}$). Furthermore,  Figure~\ref{fig:ocean-network-v2} illustrates the results specifically for the  threshold $v = \sqrt{\log(n)/n}$ (Panel A of Table~\ref{tab:direction_inference_metrics}), which reveals two additional links. The orientation of the edges reflects the estimated direction of extremal influence.

\begin{figure}[htbp!]
        \centering
        \includegraphics[width=0.4\textwidth,
                    trim={1cm 1cm 1cm 2cm},clip]{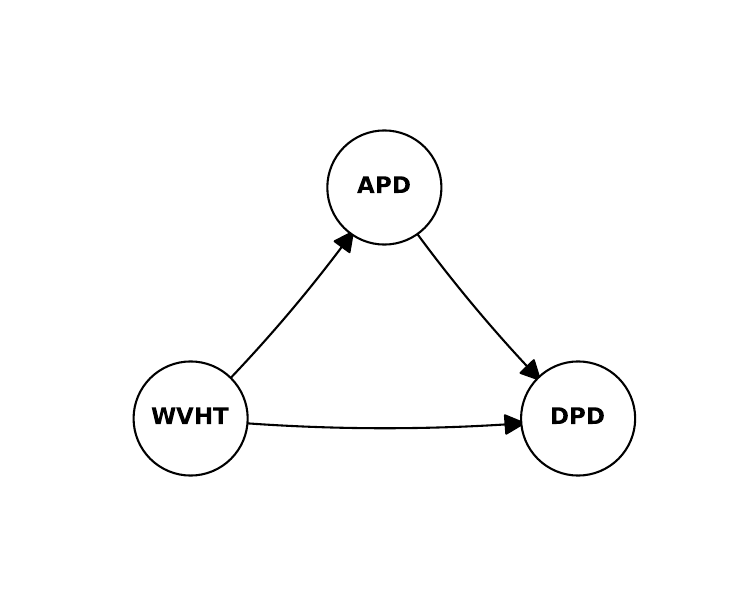} %
    \caption{Network of significant DTD links between the variables of the ocean dataset using both thresholds ($v = \sqrt{\log( n) / n}$ and $v = 1/\sqrt{n}$)}
    \label{fig:ocean-network}
\end{figure}

In particular, WVHT is tail-directed toward both APD and DPD, indicating that extremely large wave heights are more systematically associated with relatively large wave periods than the reverse configurations. This finding is consistent with energetic sea states, in which large wave heights are typically accompanied by longer wave periods, while extremely long period of waves may also occur in the absence of large heights.

In addition, a clear directional effect is detected from APD to DPD. This result suggests that extremely large values of the average wave period tend to be associated with relatively large dominant periods more frequently than extremely large dominant periods are associated with relatively large average periods. The causal interpretation of Figure~\ref{fig:ocean-network} is limited, since we don't know whether APD plays the role of a mediator between variables WVHT and DPD: is there a direct transmission channel from WVHT toward DPD or is this DTD from WVHT to DPD fully explained by the variable APD? To address this question, it would be necessary to condition the TCTEs on the variable DPD, requiring a specific asymptotic framework that is out of the scope of this paper.

At the 5\% significance level, using a threshold of $v = \sqrt{\log(n)/n}$ (as shown in Figure~\ref{fig:ocean-network-v2}), we identify two additional dependence links: WSPD $\to$ WVHT and WSPD $\to$ APD. These connections indicate that extreme wind speeds play a role in shaping both wave height and the average wave period. This aligns with physical expectations, as intense wind stress over the ocean surface directly transfers energy to the wave field, driving extreme sea states.

\begin{figure}[htbp!]
        \centering
        \includegraphics[width=0.5\textwidth,
                    trim={0cm 1cm 0cm 1cm},clip]{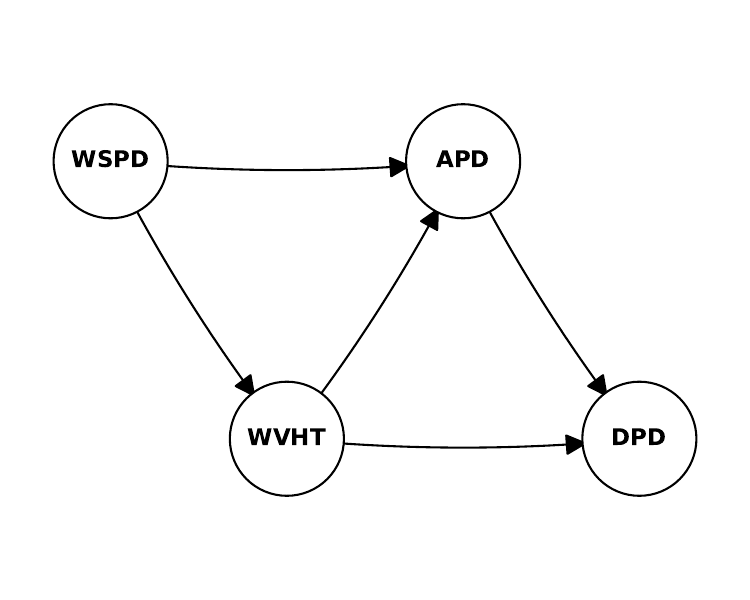} %
    \caption{Network of significant DTD links between the variables of the ocean dataset using the threshold $v = \sqrt{\log(n) / n}$. This includes two additional links compared to Fig.~\ref{fig:ocean-network} (WSPD $\to$ WVHT, WSPD $\to$ APD), indicating influence from wind speed.}
    \label{fig:ocean-network-v2}
\end{figure}

Overall, this study provides empirical evidence of directional asymmetries in oceanic extremes and illustrates the ability of the proposed measure to identify  patterns of extremal influence in multivariate environmental data.

\section{Conclusion}

This paper introduces a new measure of directional tail dependence designed to capture asymmetric extremal behavior between variables while remaining invariant to marginal distributions. Based on transformed conditional tail expectations, the proposed coefficient admits a clear interpretation in terms of directional extremal influence and can be expressed directly through the copula of the joint distribution.

We establish the consistency and asymptotic normality of a natural empirical estimator and develop a statistical testing framework for directional extremal effects. Simulation studies based on asymmetric copula constructions demonstrate good finite-sample performance and illustrate how the measure responds to varying degrees of asymmetry and tail dependence.

An application to oceanographic data highlights pronounced directional asymmetries in extreme events, revealing a coherent structure linking atmospheric forcing, wave height, and wave period characteristics.
Overall, the proposed approach provides a flexible and interpretable tool for identifying and quantifying directional features in multivariate extremes, and offers a useful complement to existing symmetric measures of tail dependence.

Beyond descriptive analysis, the directional nature of the proposed measure suggests potential for efficient use in causal discovery or causal screening procedures for extreme events.

\setlength\bibsep{0pt}
\bibliographystyle{apalike}

\addtocontents{toc}{\protect\setcounter{tocdepth}{-1}} 
\appendix  

\section{Proof of Proposition \ref{proposition:dtdc-copula}} \label{sec:proof_copula_form}

We show that $\chi^{Y\rightarrow X}(v)$ can be written in terms of copula function solely, and thus is a pure copula property independent of the marginal behaviour. We first write
$$\begin{array}{ccl}
\chi^{Y\rightarrow X}(v) & = & \E[F_X(X)|F_Y(Y)\leq v] \\
 & = & \frac{1}{\proba(F_Y(Y)\leq v)}\E[F_X(X)\indic_{\{F_Y(Y)\leq v\}}].
\end{array}$$
By continuity of the marginals, the probability integral transforms $F_X(X)$ and $F_Y(Y)$ are uniform, so that $\proba(F_Y(Y)\leq v)=v$. We have as well, $\forall x\in[0,1]$, $x=\int_0^1\indic_{\{u<x\}}du$, so that
$$\begin{array}{ccl}
\chi^{Y\rightarrow X}(v) & = & \frac{1}{v}\E\left[\int_0^1\indic_{\{u<F_X(X),F_Y(Y)\leq v\}}du\right] \\
 & = & \frac{1}{v}\int_0^1\proba(F_X(X)>u,F_Y(Y)\leq v)du.
\end{array}$$
where we used Fubini's theorem. By the law of total probability, $\proba(F_X(X)>u,F_Y(Y)\leq v)$ is equal to $\proba(F_Y(Y)\leq v)-\proba(F_X(X)\leq u,F_Y(Y)\leq v)$. Using again the uniformity of $F_X(X)$ and $F_Y(Y)$ and introducing their copula $C$ by Sklar's theorem, we conclude that
$$\begin{array}{ccl}
\chi^{Y\rightarrow X}(v) & = & \frac{1}{v}\int_0^1\left[\proba(F_Y(Y)\leq v)-\proba(F_X(X)\leq u,F_Y(Y)\leq v)\right]du \\
 & = & 1- \frac{1}{v}\int_0^1C(u,v)du.
\end{array}$$
We also obtain equation~\eqref{eq:tcte_limit_copula} as the limit of the above equation when $v\rightarrow 0$.

\section{Proof of Theorem~\ref{thm:Consistent}}\label{sec:proof_Consistent}

\begin{proof}
The estimator $(u,v)\in[0,1]^2\mapsto\widehat C_n(u,v)$ is a bounded function and it almost surely converges toward the copula $C$~\citep{Deheuvels}. As a consequence, we can apply the dominated convergence theorem on the interval $[0,1]$ and establish the conclusion regarding $\widehat \chi^{Y\rightarrow X}_n(v)$.

For the almost sure convergence of $\widehat \chi^{Y\rightarrow X}_n(v_n)$ toward $\chi^{Y\rightarrow X}$, we decompose it in two steps. First, from the uniform almost sure convergence of the empirical copula~\citep[Theorem 3.1]{Deheuvels}, we can write that for a given $\varepsilon$ there exists a threshold for $n$ above which we have almost surely
$$\begin{array}{ccl}
\left|\widehat \chi^{Y\rightarrow X}_n(v_n) - \chi^{Y\rightarrow X}_n(v_n)\right| & \leq & \frac{1}{v_n}\int_0^1\left|\widehat C_n(u,v_n) - C(u,v_n)\right|du \\
 & \leq & \frac{1}{v_n} \underset{(u,v)\in[0,1]^2}{\sup} \left|\widehat C_n(u,v) - C(u,v)\right| \\
 & \leq & \frac{\varepsilon}{2},
\end{array}$$  
because $v_n\sqrt{n/\log\log n}\rightarrow\infty$. Second, the uniform differentiability provides that, for all $n$ above another threshold,
$$\begin{array}{ccl}
\left|\chi^{Y\rightarrow X}(v_n) - \chi^{Y\rightarrow X}\right| & \leq & \int_0^1\left|\frac{C(u,v_n)}{v_n}-\partial_2 C(u,0)\right|du \\
 & \leq & \underset{u\in[0,1]}{\sup} \left|\frac{C(u,v_n)}{v_n}-\partial_2 C(u,0)\right| \\
 & = & \frac{\varepsilon}{2},
\end{array}$$
because $v_n\rightarrow 0$. Finally, these two parts give for all $n$ above a given threshold and almost surely
$$\left|\widehat \chi^{Y\rightarrow X}_n(v_n) - \chi^{Y\rightarrow X}\right| \leq \left|\widehat \chi^{Y\rightarrow X}_n(v_n) - \chi^{Y\rightarrow X}_n(v_n)\right| + \left|\chi^{Y\rightarrow X}(v_n) - \chi^{Y\rightarrow X}\right| \leq \varepsilon,$$
which is the expected result.
\end{proof}

\section{Proof of Theorem~\ref{thm:GaussChiV}}\label{sec:proof_GaussChiV}

\begin{proof}
Following Proposition 3.1 of \cite{segers2012asymptotics}, which extends a result of \cite{FermanianRadulovicWegkamp} which is slightly more restrictive regarding the regularity condition, we know that the empirical copula process $\sqrt{n}(\widehat C_n(u,v) -C(u,v))$ converges weakly towards a Gaussian process $G_C(u,v)$,
$$G_C(u,v)=B_C(u,v)-\partial_1C(u,v)B_C(u,1)-\partial_2C(u,v)B_C(1,v),$$
where $B_C$ is a Brownian bridge on $[0,1]^2$ of covariance
\begin{equation}\label{eq:CovBrownBridge}
\E\left[B_C(u,v)B_C(u',v')\right]=C(u\wedge u',v\wedge v')-C(u,v)C(u',v').
\end{equation}
Therefore, $\sqrt{n}\left(\widehat\chi^{Y\rightarrow X}_n(v)-\chi^{Y\rightarrow X}(v)\right) = \sqrt{n}\int_0^1(C(u,v) -\widehat C_n(u,v))du/v$ weakly converges toward $-\int_0^1G_C(u,v)du/v$, which is a centred Gaussian variable of variance:
$$\begin{array}{ccl}
\sigma^2_C(v) & = & \frac{1}{v^2}\int_0^1\int_0^1 \E[G_C(u,v)G_C(w,v)]dudw \\
 & = & \frac{2}{v^2}\int_0^1\int_0^w \E\left[B_C(u,v)B_C(w,v) + \partial_2C(u,v)\partial_2C(w,v)B_C(1,v)^2 \right. \\
 & & + \partial_1C(u,v)\partial_1C(w,v)B_C(u,1)B_C(w,1) \\
 & & -\partial_2C(u,v)B_C(w,v)B_C(1,v) -\partial_2C(w,v)B_C(u,v)B_C(1,v) \\
 & & -\partial_1C(u,v)B_C(w,v)B_C(u,1) -\partial_1C(w,v)B_C(u,v)B_C(w,1) \\
 & & \left.+\partial_2C(u,v)\partial_1C(w,v)B_C(1,v)B_C(w,1) + \partial_2C(w,v)\partial_1C(u,v)B_C(1,v)B_C(u,1)\right]dudw,
\end{array}$$
which leads to the expected result using the linearity of the expectation and noting that, for $u\leq w$, we have:
$$\left\{\begin{array}{ccl}
\E\left[B_C(u,v)B_C(w,v)\right] & = & C(u,v)(1-C(w,v)) \\
\E\left[B_C(1,v)^2\right] & = & v(1-v) \\
\E\left[B_C(u,1)B_C(w,1)\right] & = & u(1-w) \\
\E\left[B_C(w,v)B_C(1,v)\right] & = & C(w,v)(1-v) \\
\E\left[B_C(u,v)B_C(1,v)\right] & = & C(u,v)(1-v) \\
\E\left[B_C(w,v)B_C(u,1)\right] & = & C(u,v)-C(w,v)u \\
\E\left[B_C(u,v)B_C(w,1)\right] & = & C(u,v)(1-w) \\
\E\left[B_C(1,v)B_C(w,1)\right] & = & C(w,v) -vw \\
\E\left[B_C(1,v)B_C(u,1)\right] & = & C(u,v) -vu.
\end{array}\right.$$
\end{proof}

\section{Proof of Theorem~\ref{thm:GaussChiV_Diff}}\label{sec:proof_GaussChiV_Diff}

\begin{proof}
We use the same tools as in the proof of Theorem~\ref{thm:GaussChiV}, namely Proposition 3.1 of \cite{segers2012asymptotics}. From this, we get that 
$$\begin{array}{cl}
 & \sqrt{n}\left(\left[\widehat\chi^{Y\rightarrow X}_n(v)-\widehat\chi^{X\rightarrow Y}_n(u)\right]-\left[\chi^{Y\rightarrow X}(v)-\chi^{X\rightarrow Y}(u)\right]\right) \\
 = & \sqrt{n}\int_0^1\left(\frac{C(z,v)-\widehat C_n(z,v)}{v} - \frac{C(u,z)-\widehat C_n(u,z)}{u}\right)dz \\
 \overset{d}{\longrightarrow} & \int_0^1 \left(\frac{G_C(u,z)}{u}-\frac{G_C(z,v)}{v}\right)dz,
\end{array}$$
the limit being a centred Gaussian variable of variance:
$$\begin{array}{ccl}
\sigma^2_C(u,v) & = & \int_0^1\int_0^1 \E\left[\left(\frac{G_C(u,w)}{u}-\frac{G_C(w,v)}{v}\right)\left(\frac{G_C(u,z)}{u}-\frac{G_C(z,v)}{v}\right)\right]dwdz \\
 & = & \int_0^1\int_0^1 \left(\E\left[\frac{G_{C_P}(w,u)G_{C_P}(z,u)}{u^2}\right] - 2\E\left[\frac{G_C(u,w)G_C(z,v)}{uv}\right] + \E\left[\frac{G_C(w,v)G_C(z,v)}{v^2}\right]\right)dwdz \\
 & = & \sigma^2_{C_P}(u) + \sigma^2_C(v) - \frac{2}{uv}\int_0^1\int_0^1 \E\left[G_C(u,w)G_C(z,v)\right].
\end{array}$$
Moreover, we can decompose the integral $\int_0^1\int_0^1\E\left[G_C(u,w)G_C(z,v)\right]dwdz$ in four parts: $\int_0^u\int_0^v...dwdz$, $\int_0^u\int_v^1...dwdz$, $\int_u^1\int_0^v...dwdz$, and $\int_u^1\int_v^1...dwdz$, noted respectively $\mathcal I_1$, $\mathcal I_2$, $\mathcal I_3$, and $\mathcal I_4$. We have
$$\begin{array}{ccl}
\mathcal I_1 & = & \int_0^u\int_0^v \E\left[B_C(u,w)B_C(z,v) + \partial_2C(u,w)\partial_2C(z,v)B_C(1,w) B_C(1,v) \right. \\
 & & + \partial_1C(u,w)\partial_1C(z,v)B_C(u,1)B_C(z,1) \\
 & & -\partial_2C(u,w)B_C(z,v)B_C(1,w) -\partial_2C(z,v)B_C(u,w)B_C(1,v) \\
 & & -\partial_1C(u,w)B_C(z,v)B_C(u,1) -\partial_1C(z,v)B_C(u,w)B_C(z,1) \\
 & & \left.+\partial_2C(u,w)\partial_1C(z,v)B_C(1,w)B_C(z,1) + \partial_2C(z,v)\partial_1C(u,w)B_C(1,v)B_C(u,1)\right]dwdz \\
 & = & \int_0^u\int_0^v \left\{C(z,w)-C(u,w)C(z,v) + \partial_2C(u,w)\partial_2C(z,v)w(1-v) \right. \\
 & & + \partial_1C(u,w)\partial_1C(z,v)z(1-u) \\
 & & -\partial_2C(u,w)(C(z,w)-C(z,v)w) -\partial_2C(z,v)C(u,w)(1-v) \\
 & & -\partial_1C(u,w)C(z,v)(1-u) -\partial_1C(z,v)(C(z,w)-C(u,w)z) \\
 & & \left.+\partial_2C(u,w)\partial_1C(z,v)(C(z,w)-zw) + \partial_2C(z,v)\partial_1C(u,w)(C(u,v)-uv)\right\}dwdz,
\end{array}$$
where we used equation~\eqref{eq:CovBrownBridge} and the fact that $w\leq v$ and $z\leq u$. Similarly, considering $v\leq w$ and $z\leq u$, we obtain
$$\begin{array}{ccl}
\mathcal I_2 & = & \int_0^u\int_v^1 \left\{C(z,v)-C(u,w)C(z,v) + \partial_2C(u,w)\partial_2C(z,v)v(1-w) \right. \\
 & & + \partial_1C(u,w)\partial_1C(z,v)z(1-u) \\
 & & -\partial_2C(u,w)C(z,v)(1-w) -\partial_2C(z,v)(C(u,v)-C(u,w)v) \\
 & & -\partial_1C(u,w)C(z,v)(1-u) -\partial_1C(z,v)(C(z,w)-C(u,w)z) \\
 & & \left.+\partial_2C(u,w)\partial_1C(z,v)(C(z,w)-zw) + \partial_2C(z,v)\partial_1C(u,w)(C(u,v)-uv)\right\}dwdz.
\end{array}$$
Then, considering $w\leq v$ and $u\leq z$, we obtain
$$\begin{array}{ccl}
\mathcal I_3 & = & \int_u^1\int_0^v \left\{C(u,w)(1-C(z,v)) + \partial_2C(u,w)\partial_2C(z,v)w(1-v) \right. \\
 & & + \partial_1C(u,w)\partial_1C(z,v)u(1-z) \\
 & & -\partial_2C(u,w)(C(z,w)-C(z,v)w) -\partial_2C(z,v)C(u,w)(1-v) \\
 & & -\partial_1C(u,w)(C(u,v)-C(z,v)u) -\partial_1C(z,v)C(u,w)(1-z) \\
 & & \left.+\partial_2C(u,w)\partial_1C(z,v)(C(z,w)-zw) + \partial_2C(z,v)\partial_1C(u,w)(C(u,v)-uv)\right\}dwdz,
\end{array}$$
Finally, considering $v\leq w$ and $u\leq z$, we get
$$\begin{array}{ccl}
\mathcal I_4 & = & \int_u^1\int_v^1 \left\{C(u,v)-C(u,w)C(z,v) + \partial_2C(u,w)\partial_2C(z,v)v(1-v) \right. \\
 & & + \partial_1C(u,w)\partial_1C(z,v)u(1-z) \\
 & & -\partial_2C(u,w)C(z,v)(1-w) -\partial_2C(z,v)(C(u,v)-C(u,w)v) \\
 & & -\partial_1C(u,w)(C(u,v)-C(z,v)u) -\partial_1C(z,v)C(u,w)(1-z) \\
 & & \left.+\partial_2C(u,w)\partial_1C(z,v)(C(z,w)-zw) + \partial_2C(z,v)\partial_1C(u,w)(C(u,v)-uv)\right\}dwdz.
\end{array}$$
Putting all these formulas together leads to the expected result.
\end{proof}

\end{document}